\documentclass[11pt]{article}

\usepackage[utf8]{inputenc}
\usepackage{amsmath, amssymb,verbatim, bbm,amsfonts, hyperref, mathrsfs}
\usepackage{amsthm}
\usepackage[margin=1.111in]{geometry}
\usepackage{pifont}
\usepackage{bm}
\usepackage[makeroom]{cancel}

\newtheorem{theorem}{Theorem}

\newtheorem{definition}{Definition}
\newtheorem{lemma}{Lemma}
\newtheorem{proposition}{Proposition}

\newtheorem{problem}{Problem}

\makeatletter
\newcommand{\subalign}[1]{%
	\vcenter{%
		\Let@ \restore@math@cr \default@tag
		\baselineskip\fontdimen10 \scriptfont\tw@
		\advance\baselineskip\fontdimen12 \scriptfont\tw@
		\lineskip\thr@@\fontdimen8 \scriptfont\thr@@
		\lineskiplimit\lineskip
		\ialign{\hfil$\m@th\scriptstyle##$&$\m@th\scriptstyle{}##$\hfil\crcr
			#1\crcr
		}%
	}%
}
\makeatother

\usepackage{amssymb,latexsym}

\newcommand{\FR}{{\textup{FR}}}

\newcommand{\eqdef}{\mathop{=}\limits^{\triangle}}

\newcommand{\SRF}{\textrm{SRF}}

\newcommand{\vpi}{\boldsymbol{\pi}}

\newcommand{\wIS}{\widetilde{\IS}}

\newcommand{\transpose}[1]{{#1}^{ {\intercal} }}

\newcommand{\KMAC}{{\textup{KMAC}256}}

\newcommand{\OH}{\ket{\hh}}
\newcommand{\OG}{\ket{G}}
\newcommand{\OPI}{\ket{\pi}}

\newcommand{\ev}{{\mathbf{e}}}

\newcommand{\sv}{{\mathbf{s}}}

\newcommand{\vv}{{\mathbf{v}}}

\newcommand{\yv}{{\mathbf{y}}}

\newcommand{\Hm}{{\mathbf{H}}}

\newcommand{\Mm}{{\mathbf{M}}}

\newcommand{\para}[1]{\vspace*{-0.27cm} \textcolor{coll}{\qquad \emph{{#1}}} \vspace*{-0.27cm}}

\renewcommand{\S}{\textup{S}}

\newcommand{\FeF}{\mathfrak{Fe}_4}
\newcommand{\SHAKE}{\textup{SHAKE}{256}}

\newtheorem{statement}{Statement}

\newcommand{\wpi}{\widetilde{\pi}}

\newcommand{\Stern}{\textup{Stern}}

\newcommand{\Unif}{\xleftarrow{\$}}

\usepackage[dvipsnames]{xcolor}
\definecolor{coll}{HTML}{000090}

\usepackage{graphicx}
\usepackage{amsfonts}
\usepackage{amssymb}
\usepackage{amsmath}

\usepackage{framed}
\usepackage{latexsym}
\usepackage{breakcites}
\usepackage{color}
\usepackage{url}
\usepackage{tikz}
\usepackage{algorithm}
\usepackage{caption}
\usepackage[noend]{algpseudocode}
\usepackage{enumerate}
\usepackage{stmaryrd}
\usepackage{float}

\usepackage{thm-restate}

\usepackage{caption}
\usepackage{subcaption}
\usepackage{color}

\newtheorem{assumption}{Assumption}
\newtheorem{assumption*}{Assumption}

\newcommand{\SKeygen}{\textsc{S.keygen}}
\newcommand{\SSign}{\textsc{S.sign}}
\newcommand{\SVerify}{\textsc{S.verify}}
\newcommand{\SISKeygen}{\textsc{S}_{\IS}.\textsc{keygen}}
\newcommand{\SISSign}{\textsc{S}_{\IS}.\textsc{sign}}
\newcommand{\SISVerify}{\textsc{S}_{\IS}.\textsc{verify}}

\newcommand{\sps}[1]{{#1}\textrm{-}sp}
\renewcommand{\sp}[1]{{#1}\textrm{-}sp}
\newcommand{\spp}[1]{{#1}\textrm{-}sp+}
\newcommand{\rs}[1]{{#1}\textrm{-}rs}
\newcommand{\osp}[1]{{#1}\textrm{-}osp}

\newcommand{\E}{\mathbb{E}}

\newcommand{\ie}{\textit{i.e.}}

\renewcommand{\aa}{\mathscr{A}}
\newcommand{\bb}{\mathscr{B}}
\newcommand{\cc}{\mathscr{C}}
\newcommand{\IS}{\mathcal{IS}}
\newcommand{\FSIS}{ \mathrm{FS}^\hh[\mathcal{IS}]}

\newcommand{\Init}{\ISKeygen}
\newcommand{\Keygen}{\ISKeygen}
\newcommand{\cao}{commit-and-open}
\newcommand{\Cao}{Commit-and-open}
\newcommand{\FS}[1]{ \mathrm{FS}^\hh[#1]}
\newcommand{\hh}{ \mathcal{H}}

\renewcommand{\ss}{\mathcal{S}}

\newcommand{\ISCheck}{V_{\IS}}
\newcommand{\ISKeygen}{K_{\IS}}
\newcommand{\ISProver}{P_{\IS}}

\newcommand{\RF}{\mathcal{F}}
\newcommand{\SRRF}{\textrm{SRF}}
\newcommand{\RP}{\mathcal{P}}

\newcommand{\COMMENT}[1]{}

\newcommand{\ket}[1]{|#1\rangle}

\def\01{\{0,1\}}
\def\01{\{0,1\}}

\newcommand{\eps}{\varepsilon}

\newcommand{\zo}{\{0,1\}}

\newcommand{\cadre}[1]
{
	\begin{tabular}{|p{\textwidth}|}
		\hline
		\vspace*{-0.4cm} #1  \\
		\hline
	\end{tabular}
}

\begin{document}
	

\title{Tight quantum security of the Fiat-Shamir transform for commit-and-open identification schemes with applications to post-quantum signature schemes}
\author{{Andr{\'e} Chailloux} \\  {\small Inria de Paris, EPI COSMIQ}}
\date{}

\maketitle

	

\begin{abstract} 
	Applying the Fiat-Shamir transform on identification schemes is one of the main ways of constructing signature schemes. While the classical security of this transformation is well understood, it is only very recently that generic results for the quantum case have been proposed \cite{DFMS19,LZ19}. These results are asymptotic and therefore can't be used to derive the concrete security of these signature schemes without a significant loss in parameters.
	
	 In this paper, we show that if we start from a commit-and-open identification scheme, where the prover first commits to several strings and then as a second message opens a subset of them depending on the verifier's message, then there is a tight quantum reduction for the the Fiat-Shamir transform to special soundness notions. Our work applies to most $3$ round schemes of this form and can be used immediately to derive quantum concrete security of signature schemes.  
	 
	 We apply our techniques to several identification schemes that lead to signature schemes such as Stern's identification scheme based on coding problems, the \cite{KTX08} identification scheme based on lattice problems, the \cite{SSH11} identification schemes based on multivariate  problems, closely related to the NIST candidate MQDSS, and the PICNIC scheme based on multiparty computing problems, which is also a NIST candidate.
\end{abstract}


\section{Introduction}
Each year brings new advances in quantum technologies \cite{ABB+19} and we will soon need to deploy post-quantum cryptography in order to prevent ourselves against the potential construction of a quantum computer capable of running Shor's algorithm \cite{Sho94} and other powerful quantum algorithms. The NIST standardization process of post-quantum cryptographic primitives \cite{Nist17} (specifically encryption schemes, key encapsulation mechanisms and signature schemes) is currently ongoing and it becomes crucial to continue to build trust for these schemes. A first way to build trust is to constantly challenge the post-quantum computational assumptions by designing new quantum algorithms. Another very important aspect is to make sure we have sound security reductions even with quantum computers. In particular, several technical problems arise when translating the Random Oracle Model\footnote{In the Random Oracle Model, we model a hash function by a truly random function to which we only have black box access. This model is in all generality unrealistic and can be too strong in some pathological scenarios \cite{CGW04} but has been extremely useful for making efficient security reductions \cite{KM15} and is passing well the test of time.} (ROM) to the Quantum ROM (QROM) and we need to rewrite all the security proofs involving the QROM. \\

\para{Quantum security reductions for signature schemes} \\

In this paper, we focus on quantum security reductions for signature schemes. There are mainly $2$ families of signature schemes that use security reductions in the QROM: (1) Hash and Sign signatures and (2) signatures using the Fiat-Shamir transform on identification schemes. We understand well the security of Hash and Sign signatures in the QROM \cite{Zha12}. For those using the Fiat-Shamir transform, it is only recently that there exists a general proof of its security in the QROM \cite{DFMS19,LZ19}.

So is this the end of the story? Not quite. The results of \cite{DFMS19,LZ19} are only asymptotic and are not tight. This means that if you want your signature scheme to have $128$  bits of security, you need to choose parameters such that your post-quantum computational assumption has $256,384$ or often much more bits of security. Several schemes have tight security reductions QROM, for example those based on lossy identification schemes \cite{KLS18} or closely related \cite{ABB+19}. However, several others have only non tight security reductions and some even don't have a post-quantum security reduction, including some NIST candidates\footnote{The GeMSS signature scheme described in \cite{CFM+20} doesn't even have a full concrete security claims against classical adversaries for instance.} Of course, designers that use a non-tight security reduction could take this into account in their parameters but almost no one does this as it would be devastating for their parameters. Instead, designers often have to fix parameters as if the reductions were tight and accept not having concrete security claims. For example, in their latest design specification, the authors of PICNIC write the following: \\ \vspace*{-0.3cm}

\emph{``One caveat we note is that this generalization comes with a cost in tightness of the reduction. The reduction for the ZKB++ parameter sets looses a factor of $q^2$, and for KKW the loss is a factor $q^6$, where $q$ is the number of hash queries. As the results are non-tight, and depend on the asymptotic analysis of \cite{DFMS19}, we make no claims about the concrete security of Picnic in the QROM."} \\

\vspace*{-0.3cm}
In a similar vein, the authors of the MQDSS signature scheme \cite{CHR+20} write in their latest specifications: \\ \vspace*{-0.3cm}

\emph{``Another weakness of our security proof is that it is not at all
	tight. This is again an inherent weakness introduced by the rewinding technique of the
	forking lemma. Therefore, in order to produce a tight security reduction for MQDSS
	one would have to base the proof on different techniques. At the moment, we are not
	aware of such techniques that we could use"} \\ \vspace*{-0.3cm}

This lack of tightness can have real consequences. For example, there has been a recent attack exploiting the non-tightness of the security reduction of the MQDSS signature scheme by Kales and Zaverucha \cite{KZ19}. This was fortunately easily fixable by increasing the parameters without too much harm but this overall situation is unsettling for the trust we have in the parameter sets of these schemes, which is especially problematic since the NIST will soon choose some  post-quantum signature schemes to standardize with some fixed parameters. There is therefore an urgent need to find as tight security reductions as possible for signature schemes in the QROM. \\ \\

\para{Our work in a few words} \\

In this work, we show tight security reduction in the QROM for a large class of identification schemes: namely $3$-round {\cao}  identification schemes. We also derive a more precise reduction when considering parallel repetition of {\cao} identification schemes. We apply our results to existing signature schemes and show their concrete security, while until now, only asymptotic security was known. We consider Stern's signature scheme \cite{Ste93} \footnote{This scheme actually already has concrete quantum security bounds because the underlying identification scheme can be made lossy\cite{Lei18}. However, this introduces some losses in the parameters that don't arise with our techniques.}, the $3$ round $SSH$ signature scheme, which is a non-optimized version of the MQDSS signature the PICNIC signature scheme and the scheme from \cite{KTX08}. 

In order to find these tight reductions, we can't use rewinding techniques as they introduce non-tightness. Moreover, we have to be careful with quantum reprogramming techniques since these can also add some non-tightness as we can see from the \cite{DFMS19} results. So how do we proceed? We first extend Unruh's result and show the quantum security of the Fiat-Shamir transform for identification schemes that have some notion of soundness between statistical and computational soundness.  Then, at a crucial moment of our proof, we need to replace a random permutation by a pseudorandom permutation which is easily invertible. We use the recent result on the quantum security of Feistel networks to construct this pseudorandom permutation. We present all steps and proof techniques more in detail in Section \ref{Section:Overview}. A drawback of this work is that we use an extra computational assumption, namely the Small Range Function Instantiation Assumption which states that we can instantiate the random oracle corresponding to the commitment function with a random small range function (we will give much more details later). Despite this assumption, we still think that currently, a tight proof in the QROM with this extra assumption gives significantly more guarantees than a non-tight security proof  which can hide some real weaknesses, as shown by the attack of \cite{KZ19}.

This work is quite different and complements well the recent work \cite{DFMS19,DFM20,GHHM20} as it is more suited for  concrete quantum security claims useful for designers of signature schemes but is less general. \\

\para{Related work} \\

We briefly presented a few security results in the QROM, let us present a more detailed presentation of related work which will still be far from exhaustive. The QROM was first studied quite late actually in \cite{BDF+10} where it was correctly assessed that in the quantum setting, an adversary making queries to a random oracle should have a quantum access to it, since the hash function it models has a public description. There, they showed the security of some schemes in the QROM, as well as examples where schemes were secure in the ROM but not in the QROM. Other impossibility results showed settings where, in all generality, the quantum Fiat-Shamir transform is not secure \cite{DFG13,ARU14}. On the positive side, \cite{DFG13} proved the security of the quantum Fiat-Shamir transform when oblivious commitments are used. Unruh \cite{Unr15} then showed that it was possible to do a Fiat-Shamir like transform to remove the interaction from identification protocols. This transform is however rather inefficient and was hardly used in practice. 
More recently, there have been new positive results related to the quantum security of the Fiat-Shamir transform. If an identification scheme is lossy, then \cite{KLS18} showed tight concrete quantum security bounds for the Fiat-Shamir transform. They used this result to prove the security of the Dilithium signature \cite{DKL+17}, which is a NIST competitor.  Another related result  is the security proof of $q$TESLA \cite{ABB+19}. Unruh \cite{Unr17} showed the quantum security of the Fiat-Shamir transform for identification schemes with statistical security, or using a dual-mode hard instance generator, a property closely related to the lossiness property. Another related work is the the framework of recording quantum queries by Zhandry \cite{Zha19} which is a very powerful tool for studying random functions and the QROM.
 
Recently, $2$ papers \cite{DFMS19,LZ19} showed generic reduction for the quantum Fiat-Shamir transform. Unlike what was believed before, they show that it is actually possible to perform reprogramming of a quantum random oracle and to follow the classical proofs. Their results are not tight and lose at least a factor of $O(q^2)$ where $q$ is the number of queries to the random function. The results of \cite{LZ19} add even a larger factor of non-tightness but  can be applied to more general settings than those of \cite{DFMS19}. Then, another work \cite{DFM20} showed that this $O(q^2)$ loss is tight and showed a large of class of examples where this is necessary. We will discuss this in the next section and show that this is less harmful than it seems for security reductions. Finally a recent result \cite{GHHM20} presents optimal quantum  reprogramming techniques with applications.

\section{State of the art, overview of our results and proof techniques}\label{Section:Overview}
We will focus on the quantum security of the Fiat-Shamir transform for identification schemes and we will use known results in the QROM to transform this security into the quantum security for resulting signature schemes. As we will show, there are many cases where these reprogramming techniques still lead to a large amount of non-tightness in the proof and the goal of this paper is to present new techniques that overcomes this issue for an important class of signature schemes. \\

\para{State of the art for identification schemes} \\

In an identification scheme $\IS$, a prover $P$ has a pair of public and secret key $(pk,sk)$ and wants to convince a verifier $V$ (that sees only the public key $pk$) that he has a valid corresponding secret key $sk$. In its most standard form, an identification scheme consists of $3$ messages: a first message $x$ from $P$ to $V$, a challenge $c$ from $V$ to $P$ which is a random string and finally a response $z$ from $P$ to $V$. $V$ finally has a procedure that from $(pk,x,c,z)$ determines whether he is convinced or not.  The Fiat-Shamir transform consists of replacing the above interaction with a single message $(x,\hh(x),z)$\footnote{The message actually just consists of $(x,z)$ since $\hh(x)$ can be constructed from $x$.} from $P$ to $V$ where $\hh$ is a hash function modeled as a truly random function in the QROM.

An adversary, who knows only $pk$ and no corresponding $sk$, breaks the Fiat-Shamir transform of $\IS$ if he can construct a triplet $(x,\hh(x),z)$ that the verifier will accept. Breaking the  identification scheme (in the sense of computational soundness) means that an adversary can construct a string $x$ and, when he receives a challenge $c$, he can construct a string $z$ such that the verifier will accept $(pk,x,c,z)$.

 The security of the quantum Fiat-Shamir transform means that we can polynomially relate the above $2$ probabilities. For example, the result in \cite{DFMS19} can be stated as follows
 \begin{align}\label{Equation:S02E01}
 QADV_{\FSIS}(t,q_\hh)  \le q^2_\hh \cdot  QADV_\IS (O(t)).
 \end{align}
 On the left side is the quantum probability (or advantage) of breaking the Fiat-Shamir transform of an identification scheme $\IS$ with a quantum adversary running in time $t$ and making $q_\hh$ quantum queries to $\hh$. The right side corresponds to the probability of breaking $\IS$ for an adversary running in time $O(t).$ We can see already the term $q_\hh^2$ accounting for the non-tightness of this reduction. 
 
 There is another source of non-tightness: we often require a bound in terms of the quantum advantage for \emph{special soundness} and not computational soundness. An adversary that breaks the $2$-special soundness property is able to construct $2$ valid triplets $(x,c,z)$ and $(x,c',z')$ with $c \neq c'$\footnote{In the asymptotic case, $2$-special soundness is often defined with an efficient extractor that takes a pair of triplets and outputs a valid secret key. The current definition is similar in spirit and uses an advantage notion which is more adapted for concrete security bounds.}(the first message $x$ is the same for both triplets). This can be generalized to $\gamma$-special soundness where we require an adversary to create $\gamma$ valid triplets $(x,c_1,z_1),\dots,(x,c_\gamma,z_\gamma)$ where the challenges $c_i$ are pairwise distinct. One can relate computational soundness advantage with $\gamma$-special soundness advantage but this comes with another big loss in tightness. For example, the authors of \cite{DFMS19} use roughly\footnote{The bound is actually slightly worst as Theorem $25$ of \cite{DFMS19} (in the eprint version) generalizes Lemma $7$ of \cite{Unr12} while it should generalize Lemma $8$ in order to account for the fact that the challenges have to be pairwise distinct. The difference is however only minimal and doesn't change the asymptotic behavior, even though it may add some small dependence in the size of the challenge space.} the following bound:
 \begin{align}\label{Equation:S02E02}
 QADV_\IS(t) \le \left[QADV_\IS^{\sp{\gamma}}(O(t))\right]^{\frac{1}{2\gamma - 1}}
\end{align} 
which, when combined to Equation \ref{Equation:S02E01}, gives the  bound
 \begin{align}
QADV_{\FSIS}(t,q_\hh)  \le q^2_\hh \cdot  \left[QADV_\IS^{\sp{\gamma}}(O(t))\right]^{\frac{1}{2\gamma - 1}}.
\end{align}

We can see that already with $\gamma = 2$, we have a cubic loss in the exponent because we use special soundness and we lose a power $5$ when requiring $3$-special soundness, which are the $2$ most common cases. In conclusion, while these asymptotic results, as well as those in \cite{LZ19}, are extremely important for having post-quantum trust in the Fiat-Shamir transform for identification schemes, the amount of non-tightness is significantly too large to make concrete security claims with decent parameters. \\

\para{Overview of our results} \\

Our results will remove this non-tightness for an important class of identification schemes, namely {\cao} identification schemes. 
In a {\cao} identification scheme, the prover can extract a string $z = z_1,\dots,z_n$ from the secret key $sk.$ His first message $x = G(z_1),\dots,G(z_n)$ consists of committing to all the values $z_i$ with a commitment function $G$ and then in the second message, he reveals a subset of the $z_i$ depending on the challenge $c$. Several schemes, such as Stern's identification scheme, the Picnic identification scheme and the SSH identification scheme that inspired the MQDSS signature are of this form. They are all even more particular: they consist of a parallel repetition of a {\cao} identification scheme $\IS$ with challenge size $3$ and the advantage of the underlying post-quantum computational assumption is equal to $QADV_{\IS}^{\sp{3}}(t).$ 

Our results will rely on an additional computational, namely that the commitment function modeled as a random oracle can be replaced with a random small range function. We discuss this assumption now. 

\paragraph{What do we mean when we say $G$ can be replaced with a small range function and why is this assumption justified?} $ \ $ \\ \\
A security claim in the (Q)ROM for a signature scheme $\S$ has the following $2$ steps:
\begin{enumerate}
	\item A security proof when some hash functions are modeled with random oracles.
	\item An assumption that these random oracles can be replaced with (quantum) secure hash functions without harming the security of $\S$.
\end{enumerate}
Usually, a proof in the (Q)ROM consists of proving (1) and $\S$ is said to be secure in the (Q)ROM when it uses (2). The proof of (1) becomes useful when (2) is applied because random oracles do not exist in real life. In our case, the situation is a bit different because we need some version of (2) in order to prove (1). More precisely, we need the following statement, which we call the Small Range Function Instantiation Assumption: if the hash function used for the commitment is modeled as a random function from $\zo^x$ to $\zo^x$, then it can be replaced with a function $f = h \circ g$ where $g$ is a random function from $\zo^x$ to $[r]$ and $h$ is a random injective function from $[r]$ to $\zo^x$ without harming the security of $\S$. First, of course, this statement has to depend on the size of $r$ so we say we don't lose more that the quantum security of $f$, which here is $O(\frac{q^3}{r})$ from collision lower bounds where $q$ is the number of queries made to the commitment function. For $r$ sufficiently large, $f$ is  a very strong hash function, the only structure we have is the separate access to $h$ and $g$ and our assumption essentially says that an adversary cannot exploit this structure for breaking $\S$.
We say that this is a mild version of (2) because usually in (2), we have an explicit function $\hh$ such as $\SHAKE$ that replaces the random oracle for which we know much more structure and we assume this structure doesn't help the adversary. Recall that such an assumption (not necessarily with $\SHAKE$ but for some explicit function) has to be made in order to make the proof in the QROM useful. However, since we can't formally prove that this is a weaker assumption than using (2) with a specific function, we keep this as an extra computational assumption. 

A final remark, our assumption is on functions from $\zo^x$ to $\zo^x$ but this doesn't mean we require the signature scheme we use to use a commitment of this form. In particular, the output space can be smaller than the input space, which is not a problem for our proof.  \\

We can now go back to the statement of ours results.  Our first theorem deals specifically with the parallel repetition case.

	\begin{theorem}[Simplified]\label{Theorem:ParallelRepetitionSimple}
		Let $\IS$ be a {\cao} identification scheme that uses a commitment $G$ modeled as a random oracle. Let $\gamma \ge 2$ be an integer. For any $t,q_\hh,q_G$, and number of repetition $r$, using the Small Range Function Instantiation Assumption, we have 
		\begin{align*}QADV_{\FS{\IS^{\otimes r}}}(t,q_\hh,q_G) \le QADV_{\IS}^{\sp{\gamma}}\left(O(t)\right)+  O\left(\frac{q_\hh^2 (\gamma-1)^r}{|C|^r}\right) + O\left(\frac{q_G^3}{|M|}\right)\end{align*}
		where $|C|$ is the size of the challenge space and $|M|$ is the size of the space of each $x_i = G(z_i).$
	\end{theorem}

Before discussing what we mean by the non-standard assumption 'that can be instantiated with a random function with small range', let us present the different terms of this theorem. 
The left hand side is the probability ({\ie} advantage) that a quantum adversary has of breaking the Fiat-Shamir transform of $\IS^{\otimes r}$, the $r$-fold parallel repetition of $\IS$. The adversary is running in time $t$ and performs $q_\hh$ quantum queries to the hash function used in the Fiat-Shamir transform and $q_G$ quantum queries to the commitment function $G$. We also use the QROM and model $G$ as a truly random function. 

These terms on the right hand side are all necessary. The first term is supposed to be related to the hardness of the computational problem.  The term $O\left(\frac{q_\hh^2 (\gamma-1)^r}{|C|^r}\right)$ corresponds to applying Grover's algorithm on the challenge space. This attack appears for example in schemes that have $3$-special soundness but where an adversary can easily construct an $x$ for which he can successfully answer $2$ of the $3$ verifier's challenges. This is also the attack that was presented in \cite{DFM20} with $\gamma = 2$. So indeed, the $q^2_\hh$ might be necessary but only for the part of advantage related to the challenge attack and crucially, the $O(q_\hh^2)$ factor loss in \cite{DFMS19} isn't tight in front of the advantage to break the computational problem. What we describe here is also true for the example presented in \cite{DFM20} so their tightness result of the $O(q^2)$ loss factor is much less harmful that what it seems even for schemes where it holds.  The third term is also necessary corresponds to attacking the commitment function and breaking the binding property by finding collisions on $G$. An interesting remark about this theorem is that designers already implicitly used results very similar to Theorem \ref{Theorem:GenericSimple} but without a formal proof and used it to determine the value of $r$\footnote{For example, the PICNIC scheme is of the form $\IS^{\otimes r}$ and we have $3$-special soundness for $\IS$ (so we pick $\gamma = 3$) and $|C| = 3$. If we want $64$ bits of quantum security (so $q_\hh = 2^{64}$), we want from the challenge attack $\frac{q^2_\hh 2^r}{3^r} \le 1$ (omitting the $O(\cdot)$) which implies $r \ge 219$. If we want $128$ bits of quantum security, this $r$ has to be doubled. This corresponds exactly to the number of repetitions of the PICNIC scheme respectively for levels $1$ and $5$ of the NIST security levels.}.

 What we omitted in the description of Theorem \ref{Theorem:ParallelRepetitionSimple} is that the $O(t)$ hides some additive terms that depend on $|C|$ so they are well suited for parallel repetition of schemes with small challenge but are not suited when these are exponential. To circumvent this, we also generalize the above theorem when we don't have parallel repetition but just a single identification scheme with potentially a large challenge space. We prove the following 

\begin{theorem}[Simplified]\label{Theorem:GenericSimple}
	Let $\gamma \ge 2$ be an integer and let $\IS$ be a {\cao} identification scheme with a commitment function $G$ modeled as a random oracle that can be instantiated with a random function with small range. Using the Small Range Function Instantiation Assumption, we have for any running time $t$ and number of queries $q_\hh,q_G$
$$ QADV_{\FS{\IS}}(t,q_\hh,q_G) \le QADV_{\IS}^{\osp{\gamma}}\left(O(t)\right) + O\left(\frac{q_\hh^2 \gamma}{|C|}\right) + O\left(\frac{q_G^3}{|M|}\right).$$
\end{theorem}
Here, the $O(\cdot)$ terms do not depend on $|C|$ anymore. 
This theorem is very similar to Theorem \ref{Theorem:ParallelRepetitionSimple} but the reduction is to a weaker notion of special soundness, namely output special soundness (hence the $\osp{\gamma}$ in the theorem) that we will discuss more in detail in the paper.  Informally, we want again the adversary to produce $\gamma$ valid triplets $(x,c_i,z_i)$ except that he doesn't need to know what are the challenges $c_i$ that correspond to the $z_i.$ The identification schemes we study  all can use Theorem \ref{Theorem:ParallelRepetitionSimple} but it would be interesting to see if some other schemes could use Theorem \ref{Theorem:GenericSimple}. \\

\setcounter{theorem}{0}
Finally, an important conceptual step of our results is to relate the quantum Fiat-Shamir advantage for any identification scheme (so not necessarily commit and open) to the notion of $\gamma$-rigid soundness. This notion can be seen as a computational-statistical notion of soundness meaning that the adversary is computationally bounded when producing the first message $x$ but unbounded when producing the response $z$ (that depends on $pk,x,$ and the challenge $c$). Informally, an adversary breaks the $\gamma$-rigid soundness property if he can construct $x$ such that he will be able to answer in a valid way at least $\gamma$ different challenges (he is unbounded for this second message). We prove the following

\begin{proposition}\label{Proposition:GenericSimple}
	For any integer $\gamma \ge 2$, time $t$, number of queries $q_\hh$, and identification scheme $\IS$, we have
	$$ QADV_{\FS{\IS}}(t,q_\hh) \le QADV_{\IS}^{\rs{\gamma}}(t,q_\hh) + O(\frac{q^2_\hh \gamma}{|C|}).$$
\end{proposition}

This proposition can be seen as a generalization of Unruh's reduction from the quantum advantage of the Fiat-Shamir transform to statistical soundness. The fact that we impose a $\gamma$ threshold here in our rigid soundness definition makes it easier to related to $\gamma$-special soundness without any losses in tightness. We use this proposition for {\cao} identification schemes but it could have more applications. \\

\para{Techniques used} \\ 

How to we achieve our results? The most common ways of proving the quantum security of the Fiat-Shamir transform use techniques such as quantum rewinding or quantum reprogramming. These techniques are very general but introduce some non-tightness that we want to avoid so we have to manage without them. Our starting point is to use Unruh's result on the quantum security of the Fiat-Shamir transform when the underlying identification scheme has statistical soundness. In this case, things are fairly easy and we can invoke quantum lower bounds on the search problem to conclude. As we wrote above, we first introduce the notion of $\gamma$-rigid soundness to achieve Proposition \ref{Proposition:GenericSimple} that holds for any identification scheme.

We then look more precisely at {\cao} identification schemes, where during the first message, the prover commits to some values $x = G(z_1),\dots,G(z_n)$ where $G$ is the commitment function and reveals a subset of those $z_i$ as his second message. We first show that we can replace this function $G$ with a random permutation $\sigma$ \footnote{We note here that this replacement is just part of a proof technique. We prove the security of identification schemes for random commitment schemes which are not permutations.}. This comes from the fact that the actual values of $G(z_1),\dots,G(z_n)$ are used only for computing the challenge $c = \hh(G(z_1),\dots,G(z_n))$. Since $\hh$ is also random, we show that this change of $G$ doesn't change the quantum advantage, on average on $\hh.$ This is actually where we use the Small Range Function Instantiation Assumption and we don't use later in the proof. 

However, because we want tight results, we are far from done. We can't use generic relations from computational soundness to $\gamma$-special soundness (like the one in Equation \ref{Equation:S02E02}). We need to directly reduce to $\gamma$-special soundness without going through computational soundness. To do so, we need from the string $\sigma(z_1),\dots,\sigma(z_n)$ to be able to recover the whole string $z = z_1,\dots,z_n$. However, we only have black box access to $\sigma$ and we don't have access to a inversion oracle. The idea we use to do this is to replace $\sigma$ with a random permutation from a pseudorandom permutation family $\{\widetilde{\pi}_K\}$ which doesn't change the security claim but which is easily invertible. From there, we can tightly relate the Fiat-Shamir advantage to a $\gamma$-special soundness advantage.  How do we construct this function $\widetilde{\pi}_K$? We use recent results on the quantum security of Feistel networks from \cite{HI19}. This result shows how to construct quantum secure random permutations from random functions with black box access. These Feistel networks also have the property that they are easily computable \emph{and invertible}, even when the underlying random function is hard for the preimage finding problem. We use as the underlying pseudorandom function the keyed $\KMAC$ functions, which are believed to be quantum secure.

Putting this all together, we can relate the quantum Fiat-Shamir to special soundness notions. In order to use the security of the Feistel networks, we have to artificially increase the size of the input space of the commitment scheme and we also replaced the random function $G$ with a function $\wpi_K$ for random $K$. So how we can conclude about special soundness for the original scheme. For Theorem \ref{Theorem:GenericSimple}, this is immediate as our transformations do not change the $\gamma$-output special soundness advantage. However, this is not true for $\gamma$-special soundness. For Theorem \ref{Theorem:ParallelRepetitionSimple}, we actually reduce to a stronger variant of $\gamma$-special soundness which is also invariant under our transformations which immediately implies Theorem \ref{Theorem:ParallelRepetitionSimple}. 

We now dive in the more formal part of this paper.

\section{Preliminaries}\label{Section:Preliminaries}
\paragraph{Basic notations.}
For an integer $N \in \mathbb{N}^*$, we denote by $[N]$ the set $\{1,\dots,N\}.$ For a (usually probabilistic) algorithm $\mathrm{A}(\cdot)$, $x \leftarrow \mathrm{A}(\cdot)$ means that we run $\mathrm{A}(\cdot)$ with some fresh randomness and get some output $x$. We will sometimes also use the notation $\mathrm{A}(\cdot) \rightarrow x.$ We will also use the notation $x \leftarrow D$ when $D$ is a distribution when we sample $x$ from $D.$ For a set $S$, the notation $x \Unif S$ means that $x$ is chosen uniformly at random from the set $S$. Let $\RF^X_Y$ be the set of functions from $X$ to $Y$ and let $\RP^X$ be the set of permutations acting on $X$. The notation $\eqdef$ designs an equality which is a definition. We will use Landau notations and also $O_n(\cdot)$ meaning $poly(n) O(\cdot)$.
\subsection{Quantum query algorithms.}
\indent In this work, we will often work with query algorithms that have a black box access to some deterministic function $f$. A classical access to $f$ means that we can perform queries that on input $x$ outputs $f(x)$. A quantum access to $f$ means that we can perform the unitary $U_f$ in a black box manner, where
$$ U_f : \ket{x}\ket{y} \rightarrow \ket{x}\ket{y \oplus f(x)}.$$

A quantum query algorithm with classical access to $f$ will be denoted $\aa^f$ and a quantum query algorithm with quantum access to $f$ will be denoted $\aa^{\ket{f}}.$ For any quantum algorithm $\aa$, we denote by $|\aa|$ it's total running time. We write $|\aa^{\ket{f}}| = (t,q_f)$ when $\aa^{\ket{f}}$ runs in time $t$ and performs $q_f$ quantum queries to $f.$  We can also write $|\aa^{\ket{f}}| = (*,q_f)$ to specify only the number of queries but not the running time. Unless stated otherwise, black box calls to $f$ or $U_f$ are efficient and we fix the running time of a query to be equal to $1.$

In the notation $\aa^{\ket{f}}$, the behavior of the query algorithm is described by $\aa$ and the superscript $\ket{f}$ only indicates which function is queried. This means that the algorithm $\aa^{\ket{g}}$ behaves exactly as $\aa^{\ket{f}}$ where calls to $U_f$ are replaced with calls to $U_g.$ We can also write $\aa$ for a quantum query algorithm where the queried function is not specified.

A query algorithm can perform queries to different functions. For example $\aa^{\ket{f},\ket{g},\ket{h}}$ has a black box access to the $3$ unitaries $U_f,U_g,U_h.$ We write $|\aa^{\ket{f},\ket{g},\ket{h}}| = (t,q_f,q_g,q_h)$ to denote the fact that $\aa^{\ket{f},\ket{g},\ket{h}}$ runs in time $t$, performs $q_f$ queries to $U_f$, $q_g$ queries to $U_g$ and $q_h$ queries to $U_h.$ Finally, we define the $q$-query quantum variational distance between $2$ distributions $D_1,D_2$ on functions as 
$$ \Delta_q(D_1,D_2) \eqdef \max_{\aa : |\aa| = (*,q)} \left| \Pr_{f \leftarrow D_1}[\aa^{\ket{f}}(\cdot) \textrm{ outputs } 0] -  \Pr_{g \leftarrow D_2}[\aa^{\ket{g}}(\cdot) \textrm{ outputs } 0] \right|.$$

\subsection{Hash functions and Feistel networks}
\paragraph{$\SHAKE$.} A prime function for instantiating random oracles in the post-quantum setting is $\SHAKE$. It is a SHA-3 variant \cite{BDPV11} that uses the sponge construction with variable input and output sizes. We write $\SHAKE_{X,Y}$ to explicit the input space $X$ and output space $Y$. The sponge construction is known to be quantum secure \cite{CHS19} and it is standard in the QROM to model $\SHAKE_{X,Y}$ with a random function in $\RF^X_Y$ for which we only have black box access. There are keyed versions of $\SHAKE$ called $\KMAC_K$ indexed with a key $K$. This function family is believed to be a quantum secure pseudo random function family.

\paragraph{Feistel networks.}Feistel networks are a generic way to transform pseudorandom functions in pseudorandom permutations. They were first studied by Luby and Rackoff \cite{LR88}, and we know well their classical security. Recently, the quantum security was proven for $4$ round Feistel networks. Very briefly, the $4$ round Feistel network starts with a function $f \in \RF^{\zo^n}_{\zo^n}$ and constructs a permutation $\FeF(f) \in \RP^{\zo^{2n}}.$ $\FeF(f)$ uses $4$ black box calls to $f$ and both $\FeF(f)$ and $\FeF(f)^{-1}$ are efficiently computable if we know how to efficiently compute $f$ (but not necessarily $f^{-1}$). The quantum security of $\FeF$ was recently proven in \cite{HI19}:
\begin{proposition}[\cite{HI19}]\label{Proposition:Feistel}
	Let $D_1$ be the distribution sampled as follows: $f \Unif \RF^{\zo^n}_{\zo^n}, \textrm{ return } \FeF(f).$  We have
	$ \Delta_q(D_1,\RP^{\zo^{2n}}) \le O(\sqrt{\frac{q^6}{2^{n}}}).$
\end{proposition}

\subsection{Quantum lower bounds}\label{Section:SmallRangeFunctions}
We will use a generalization of Grover's lower bound for the search problem. 
\begin{lemma}\label{Lemma:Search}
	Let $X$ and $Y$ be respectively an input set and an output set.  For each $x \in X$, we associate a set $U_x \subseteq Y$ such that $\frac{|U_x|}{|Y|} \le \eps$. For any quantum query algorithm $\aa$ with $|\aa| = (*,q)$, we have 
	$$\Pr[\hh(x) \in U_x :  \hh \Unif \RF^X_Y, x \leftarrow \aa^{\OH}(\cdot)] \le O(q^2\eps).$$
\end{lemma}
The above lemma was implicitly stated and proven in \cite[Theorem21]{Unr17}. Another lower bound that we will use is Zhandry's quantum lower bound on distinguishing a random permutation
from a random function with small range \cite{Zha15}. We fix a set $X$, an integer $r$ such that $[r] \subseteq X$, and define the following distribution $\SRRF^X_r$ on functions in $\RF^X_X$, which can be sampled as follows: 
\begin{itemize} \setlength\itemsep{-0.2em}
	\item Draw a random function $g \Unif \RF^{X}_{[r]}.$	
	\item Draw a random injective function h from $[r]$ to $X$.
	\item Output $h \circ g.$
\end{itemize}
Notice that since we imposed $[r] \subseteq X$, we can consider $g$ as an element of $\RF^X_X$ and choose for $h$ a random permutation in $\RP^X$ which will lead to the same distribution. Also, we can replace $[r]$ with any other set $Y \subseteq X$ with $|Y| = r.$  Zhandry’s lower bound can be stated as follows:

\begin{proposition}[\cite{Zha15}]\label{Proposition:ZhandryLB}
	$\Delta_q(\SRRF^X_r,\RP^X) \le O(\frac{q^3}{r}).$
\end{proposition}

 \subsection{The (quantum) random oracle model, and the Small Range Function Instantiation Assumption}
 The Random Oracle Model (ROM) is a strong model where make the assumption that one or several hash functions - which are deterministic and have an explicit description - used in a cryptographic primitive can be modeled as truly random function $\hh$ with only black box access. In the Quantum Random Oracle Model (QROM), we have a quantum black box access to this function meaning we give only access to the unitary 
 $$U_\hh : \ket{x}\ket{y} \rightarrow \ket{x}\ket{y \oplus \hh(x)}.$$
  The (Q)ROM is a quite strong assumption since an explicit deterministic hash function cannot be in all generality a truly random function. In order to prove the security of a signature scheme $\S$ in the (Q)ROM, we need the following statement and assumption
 \begin{statement}\label{Statement:QROM}
 	$\S$ is secure when $1$ or several (hash) functions used in the scheme are modeled as uniformly random functions to which we only have (quantum) black box access which we call random oracles. 
 \end{statement}
\begin{assumption}\label{Assumption:QROM}
	These random oracles can be instantiated ({\ie } replaced) with suitable secure hash functions without harming the security of,$\S$ whose descriptions are publicly known by the adversary.
\end{assumption}

In all generality Assumption \ref{Assumption:QROM} is too strong. First, for any specific hash function $h_0$, it is fairly easy to construct a scheme which is secure with a random oracle but insecure for this function\footnote{\label{footnote:13}We reproduce here the example from \cite{KL14}, exercise 13.2. Consider a signature scheme $\Pi$ that is secure in the standard model and consider the signature $\Pi'$ that outputs the secret key if $\hh(0) = h_0(0)$ and that uses $\Pi$ otherwise. $\Pi'$ will be secure when $\hh$ is modeled as a random oracle but not when it is instantiated with $h_0$.} Moreover, it was shown \cite{CGW04} that it is actually possible to construct a signature scheme $\S$ for which Statement \ref{Statement:QROM} holds but is insecure for \emph{any instantiation} of the random oracle. This result put serious doubt on the validity of the ROM. However, more than $15$ years after this result, there has been no attack on a used cryptographic scheme.

We assume now that we are not in a pathological case\footnote{The term pathological was actually used in this setting by the authors of \cite{CGW04}.}  where the impossibility result of \cite{CGW04} applies. For what hash functions do we expect Assumption \ref{Assumption:QROM} to hold? There is actually no clear answer to this question. In practice, we have that a hash function $h_0$ is a suitable instantiation of the (Q)RO if it meets the following criteria:
\begin{enumerate}
	\item $h_0$ is a secure cryptographic hash function, ${\ie}$ it is preimage resistant, second preimage resistant and collision resistant.
	\item $h_0$ is constructed independently of the rest of the scheme. This vague statement is to avoid such attacks as those presented in Footnote \ref{footnote:13}
\end{enumerate}

Usually, proofs in the ROM prove Statement \ref{Statement:QROM} and do not care about the instantiation. In our work, we need something similar to Assumption \ref{Assumption:QROM} to prove Statement $\ref{Statement:QROM}$ with a tight security reduction for signatures based on {\cao} identification schemes. This is quite non-standard and weaker than a standalone proof of Statement \ref{Statement:QROM} but since in practice, we need Assumption \ref{Assumption:QROM} for the ROM proof to be useful, we argue that this should not harm the practical security of signature schemes for which our results give a tight security proof. More precisely, we will require the following assumption:

\begin{assumption}[Small Range Function Instantiation Assumption]\label{Assumption:QROM2}
	A random oracle from $X$ to $X$ can be replaced with a a random function $f  = (h \circ g) \Unif \SRF^X_r$  without harming too much the the studied scheme (for $r$ large enough), whether its the signature scheme itself or the underlying identification scheme, even if we give access to the structure of $f$, {\ie } quantum black box access to $h$ and $g$.
	
	If $q$ is the number of queries to the random oracle then this security loss should be at most $O(\frac{q^3}{r})$ when we allow $q$ queries to $h$ and $g$.
\end{assumption}

Notice that if we didn't give access to $h$ and $g$, this would just be Proposition \ref{Proposition:ZhandryLB}. Recall that a function $f \Unif \SRF^X_r$ can be written $f = h \circ g$ where $g$ is a random function from $X$ to $[r]$ and $h$ is a random injective function from $[r]$ to $X$, such that we have (quantum) black box access to these $2$ functions. Assumption \ref{Assumption:QROM2} should hold when the adversary knows the description of $f$ which in this case means that he has (quantum) black box access to $h$ and $g$. 
 This assumption is quite milder than Assumption \ref{Assumption:QROM} since $f$ here has very little exploitable structure (only the black box access to $g$ and $h$), which is much less than the structure that exists in explicit deterministic functions. We also argue it has all the properties of a good hash function: it has $r/2$ bits of security against preimage and second-preimage attacks and $r/3$ bits of security against collision attacks\footnote{The preimage and collision security comes from the corresponding security of $g$ while the second-preimage resistance of $f$ comes from the corresponding resistance of $h$ and the injectivity of $g$.}. Moreover, the choice of $h,g$ is random and independent of the signature scheme so a function $f \Unif \SRF_r$ seems like a really secure instantiation choice, as long as we take $r \ge 3\lambda$ where $\lambda$ is the desired number of security bits. 

Another argument for Assumption \ref{Assumption:QROM2} is that if it doesn't hold for a signature scheme $\S$, then it seems arguably hard to trust $\S$ where the RO is instantiated with an explicit function while if we want a proof in the (Q)ROM to be useful for proving practical security, we need to have this trust.

\section{Identification schemes}\label{Section:IdentificationSchemes}
\subsection{First definitions}
An identification scheme $\IS = (\ISKeygen,\ISProver,\ISCheck;M,C,R)$, consists of the following:
\begin{itemize} \setlength\itemsep{-0.2em}
	\item A key generation algorithm $\ISKeygen(1^\lambda) \rightarrow (pk,sk)$.
	\item The prover's algorithm $\ISProver = (P_1,P_2)$ for constructing his messages. We have $P_1(sk) \rightarrow (x,St)$ where $x \in M$ corresponds to the first message and $St$ is some internal state. $P_2(sk,x,c,St) \rightarrow z$ where $c \in C$ is the challenge from the verifier and $z \in R$ the prover's response (second message).
	\item A verification function $\ISCheck(pk,x,c,z)$ used by the verifier that outputs a bit, $0$ corresponds to `Reject' and $1$ to `Accept'.
\end{itemize}
Notice that we  specify in the description of $\IS$ the sets $M,C,R$ corresponding respectively to the first message space, the challenge space and the second message ({\ie} response) space. All the different algorithms presented above are efficient and we will usually omit their running times ({\ie} fix them to $1$), in order to reduce the amount of notations we introduce.  Even though we deal with concrete security parameters in this paper, we keep the notation $\Keygen(1^\lambda)$ with a unary representation of a security parameter $\lambda$ to remind this implicit efficiency requirement. \\ 
We present below more precisely the different steps of an identification scheme. 
\begin{center}\cadre{\begin{center}
		Identification scheme $\IS = (\ISKeygen,\ISProver = (P_1,P_2),\ISCheck;M,C,R)$
	\end{center}
	\textbf{Initialization.} $(pk,sk) \leftarrow \ISKeygen(1^\lambda)$. The prover has $(pk,sk)$ and the verifier $pk$. \\
	\textbf{Interaction.}
	\begin{enumerate}
		\setlength\itemsep{-0.2em}
		\item The prover generates $(x,St) \leftarrow P_1(sk)$ and sends $x \in M$ to the verifier.
		\item The verifier picks $c \Unif C$ and sends $c$ to the prover.
		\item The prover generates $z \leftarrow P_2(sk,x,c,St)$ and sends $z \in R$ to the verifier.
	\end{enumerate}
	\textbf{Verification.} The verifier accepts iff. $\ISCheck(pk,x,c,z) = 1$. \\
}\end{center}
$ \ $ \\ 

We denote by $\IS^{\otimes r}$ the $r$-fold parallel repetition of $\IS$, which consists of the following 
\begin{center} \cadre{\begin{center}
		Identification scheme $\IS^{\otimes r}$ when $\IS = (\ISKeygen,\ISProver = (P_1,P_2),\ISCheck;M,C,R)$
	\end{center}
	\textbf{Initialization.} $(pk,sk) \leftarrow \ISKeygen(1^\lambda)$. The prover $P$ has $(pk,sk)$ and the verifier $V$ has $pk$. \\
	\textbf{Interaction.}
	\begin{enumerate}
		\setlength\itemsep{-0.2em}
		\item $P$ generates $(x^1,St^1),\dots,(x^r,St^r)$ where for each $i \in [r]$, he generates $(x^i,St^i) \leftarrow P_1(sk)$. He then sends $x = x^1,\dots,x^r$ to $V$.
		\item $V$ picks a random $c = c^1,\dots,c^r$ where each $c^i \Unif C$ and sends $c$ to $P$.
		\item $P$ generates $z = (z^1,\dots,z^r)$ where for each $i \in [r]$, $z^i \leftarrow P_2(sk,x^i,c^i,St^i)$ and sends $z$ to $V$.
	\end{enumerate}
	\textbf{Verification.} The verifier $V$ accepts iff. $\forall i \in [r], \ISCheck(pk,x^i,c^i,z^i) = 1$. 
} \end{center} 

Now, let's present the properties we want an identification scheme to verify. The first property we want from an identification scheme is that the verifier accepts if a prover runs the scheme honestly.

\begin{definition}[Completeness]
	An identification scheme $\IS = (\ISKeygen,\ISProver = (P_1,P_2),\ISCheck;M,C,R)$ has perfect completeness if
	\begin{align*}
	\Pr\left[\ISCheck(pk,x,c,z) = 1 \left| \substack{(pk,sk) \leftarrow \ISKeygen(1^\lambda) \\ (x,St) \leftarrow P_1(sk) \\ c \Unif C \\ z \leftarrow P_2(sk,x,c,St) = 1} \right.\right] = 1.
	\end{align*}
\end{definition}

We only consider here perfect completeness but almost perfect completeness where the probability above is very close to $1$  could also be used.

The second property we want is honest-verifier zero-knowledge, meaning that an honest verifier cannot extract any information (in particular about the secret key $sk$), from its interaction with an honest prover. 

\begin{definition}[HVZK]
	An identification scheme $\IS = (\ISKeygen,\ISProver,\ISCheck;M,C,R)$ is $\eps$-HVZK if there exists an efficient simulator $Sim$ such that the $2$ distributions $D_1$ and $D_2$ sampled as follows:
	\begin{itemize}
		\setlength\itemsep{-0.2em}
		\item $D_1 : (pk,sk) \leftarrow \Init(1^\lambda), \ (x,St) \leftarrow P_1(sk), c \xleftarrow{\$} C, z \leftarrow P_2(sk,x,c,St)$, return $(x,c,z)$,
		\item $D_2 : (pk,sk) \leftarrow \Init(1^\lambda), \ (x',c',z') \leftarrow Sim(pk,1^\lambda)$, return $(x',c',z')$,
	\end{itemize}
	have statistical distance\footnote{The statistical distance between $2$ distributions is defined as $\Delta(D_1,D_2) \eqdef \frac{1}{2} \sum_{y} \large|\Pr_{x \leftarrow D_1}[x = y] - \Pr_{x \leftarrow D_2}[x = y]\large|.$}  at most $\eps$.
\end{definition}
Finally, the third property that we require is soundness. We don't want an efficient cheating prover that doesn't know the secret key $sk$ to make the verifier accept. There are different notions of soundness and the interplay between them will play an important role in our proofs. 

\paragraph{Different flavors of soundness.}

We provide here notions of soundness in terms of advantage, which are well suited  when dealing with concrete security bounds. We first define the notion of (computational) soundness advantage for a quantum cheating adversary $\aa$.
\begin{definition}[Quantum soundness advantage]
	Let $\IS = (\ISKeygen,\ISProver,\ISCheck;M,C,R)$ be an identification scheme. For any quantum algorithm (a quantum cheating prover) $\aa = (\aa_1,\aa_2)$, we define
	\begin{align*}
	QADV_{\IS}(\aa) \eqdef \Pr\left[\ISCheck(pk,x,c,z) = 1 \left| \substack{(pk,sk) \leftarrow \Init(1^\lambda) \\  (x,St) \leftarrow \aa_1(pk) \\ c \Unif C \\ z \leftarrow \aa_2(pk,x,c,St)}\right.\right]
	\end{align*}
	and $QADV_{\IS}(t) \eqdef \max_{\substack{\aa = (\aa_1,\aa_2), \\  |\aa_1| + |\aa_2| = t}} \left(QADV_{\IS}(\aa)\right).$
\end{definition}

In the context of identification schemes, we define the quantum $2$-special soundness advantage as follows

\begin{definition}
	Let $\IS = (\ISKeygen,\ISProver,\ISCheck;M,C,R)$ be an identification scheme. For any quantum algorithm $\aa$, we define
	\begin{align*}
	QADV_{\IS}^{\sps{2}}(\aa) \eqdef  \Pr\left[\ISCheck(pk,x,c,z) = 1 \wedge \ISCheck(pk,x,c',z') = 1 \wedge c \neq c' \left| \substack{(pk,sk) \leftarrow \Init(1^\lambda) \\ (x,c,z,c',z') \leftarrow \aa(pk)}\right.\right] 
	\end{align*}
	and $
	QADV_{\IS}^{\sps{2}}(t) \eqdef \max_{\aa : |\aa| = t} \left(QADV_{\IS}^{\sps{2}}(\aa)\right).$
\end{definition}

A small $2$-special soundness advantage means that it is hard for a quantum adversary to construct $2$ valid transcripts $(x,c,z)$ and $(x,c',z')$ with $c \neq c'.$ 
This notion can be extended to $\gamma$-special soundness, where we require more than $2$ transcripts.

\begin{definition}
	Let $\IS = (\ISKeygen,\ISProver,\ISCheck;M,C,R)$ be an identification scheme. For any quantum algorithm $\aa$, we define
	\begin{multline*}
	QADV_{\IS}^{\sps{\gamma}}(\aa) \eqdef \Pr\Big[\forall j \in [\gamma], \ 
	\ISCheck(pk,x,c_j,z_j) = 1 \ \wedge \\ \left(c_1,\dots,c_\gamma \textrm{ are pairwise distinct}\right) \Big| \substack{(pk,sk) \leftarrow \Init(1^\lambda) \\ (x,c_1,\dots,c_\gamma,z_1,\dots,z_\gamma) \leftarrow \aa(pk)}
	\Big]
	\end{multline*}
	and $
	QADV_{\IS}^{\sps{\gamma}}(t) \eqdef \max_{\aa : |\aa| = t} \left(QADV_{\IS}^{\sps{\gamma}}(\aa)\right).$
\end{definition}

\subsection{The Fiat-Shamir transform for identification schemes}\label{Section:Fiat-Shamir}

The Fiat-Shamir transform \cite{FS86} is a major cryptographic construction that converts any $\Sigma$-protocol, in our case any identification scheme into an non-interactive protocol. The idea is to use a hash function $\hh : M \rightarrow C$, and to replace the verifier's challenge $c \in C$ by the string $\hh(x)$ where $x$ is the prover's first message. Since the prover can compute $\hh(x)$ himself, there is no need for interaction anymore. For any identification scheme $\IS$, we denote by $\FSIS$ its Fiat-Shamir transform, for a fixed function $\hh.$ \\ \\
\cadre{\begin{center}
		Running $\FSIS$ for an identification scheme $\IS = (\ISKeygen,\ISProver,\ISCheck;M,C,R)$
	\end{center}
	\textbf{Initialization.} $(pk,sk) \leftarrow \Keygen(1^\lambda)$. The prover $P$ has $(pk,sk)$ and the verifier $V$ has $pk$. \\
	\textbf{One-way communication.} $P$ generates $(x,St) \leftarrow P_1(sk)$, computes $c = \hh(x)$ and generates $z \leftarrow P_2(sk,x,c,St)$. He sends the pair $(x,z)$ to the verifier. \\
	\textbf{Verification.} The verifier accepts iff. $\ISCheck(pk,x,\hh(x),z) = 1$. \\
} $ \ $ \\

The Fiat-Shamir transform is very useful as it can be used (among other things) to construct signature schemes from identification schemes. The quantum Fiat-Shamir advantage for $\FSIS$ is defined as follows:

\begin{definition}
	Let $\IS = (\ISKeygen,\ISProver,\ISCheck;M,C,R)$ be an identification scheme and $\FSIS$ its Fiat-Shamir transform. Let $\aa$ be a quantum query algorithm. We define 
	\begin{align*} QADV_{\FSIS}(\aa^{\OH}) & \eqdef  \Pr\left[V(x,\hh(x),z) = 1 \left| \substack{(pk,sk) \leftarrow \Init(1^\lambda) \\ (x,z) \leftarrow \aa^{\OH}(pk)}\right.\right] \\
	& \end{align*} 
	and $QADV_{\FSIS}(t,q_\hh)  \eqdef \max_{\aa : |\aa| = (t,q_\hh)} \left(QADV_{\FSIS}(\aa^{\OH})\right).$
\end{definition} 

In the QROM, this function $\hh$ is modeled as a random function to which we only have black box access. In this model, the quantum Fiat-Shamir advantage that we are interested in is 
$$ \E_{\hh \Unif \RF^{M}_C} \left( QADV_{\FSIS}(\aa^{\OH}) \right).$$

\subsection{Signature schemes}
All our technical work is on identification scheme but the finality is to prove the security of signature schemes. We discuss signature schemes and how the security of identification schemes implies the security of signature schemes in Appendix \ref{Appendix:Signatures}.

\subsection{Relating the quantum Fiat-Shamir security to rigid soundness}\label{Section:RigidSoundness}
In this section, we introduce the notion of rigid soundness and relate the quantum Fiat-Shamir security of any identification scheme to this notion. 
Throughout this section, we fix an identification scheme $\IS = (\ISKeygen,\ISProver,\ISCheck;M,C,R).$ 
 We first define the set $VC^\IS_x$ of valid challenges for $x \in M$ as well as the set $VC_{\ge \gamma}^\IS$ of elements having at least $\gamma$ valid challenges for any $\gamma \in \mathbb{N}$:
 \begin{align*}
 VC^\IS_x \eqdef \{c \in C: \exists z \in R, \ISCheck(pk,x,c,z) = 1\} \quad ; \quad VC^\IS_{\ge \gamma} \eqdef \{x \in M : |VC^\IS_x| \ge \gamma\}.
 \end{align*}
 We can now define the quantum $\gamma$-rigid soundness advantage for a quantum algorithm $\aa$ as follows:
 \begin{definition}[Quantum $\gamma$-rigid soundness advantage]
$$
QADV_{\IS}^{\rs{\gamma}}(\aa) \eqdef \Pr\left[x \in VC^{\IS}_\gamma \left| \ \substack{(pk,sk) \leftarrow \Init(1^\lambda) \\ x \leftarrow \aa(pk)} \right.\right] \ ; \ QADV_{\IS}^{\rs{\gamma}}(t) \eqdef \max_{\aa : |\aa| = t} QADV_{\IS}^{\rs{\gamma}}(\aa).$$
\end{definition}
\noindent We now relate the security of the Fiat-Shamir transform to a rigid soundness advantage.

\begin{proposition}\label{Proposition:Generic}
	For any query algorithm $\aa^{\OH}$ with $|\aa^{\OH}| = (t,q_\hh),$ for any integer $\gamma \ge 2$, we have
	$$ \E_{\hh \leftarrow \RF^{M}_C} \left[QADV_{\FS{\IS}}(\aa^{\OH})\right]  \le QADV_{\IS}^{\rs{\gamma}}(t,q_\hh)+ O(\frac{q^2_\hh \gamma}{|C|}).$$
\end{proposition}
\begin{proof}
	Fix a query algorithm $\aa^{\ket{\hh}}$ with $|\aa^{\OH}| = (t,q_\hh)$ and an integer $\gamma \ge 2$. 
	\begin{align}\label{Eq:S06E01}
	\E_{\hh \leftarrow \RF^{M}_C} [QADV_{\FS{\IS}}(\aa^{\OH})] & =  \Pr_{\hh \leftarrow \RF^{M}_C}\left[\ISCheck(pk,x,\hh(x),z) = 1  \ \left| \ \substack{(pk,sk) \leftarrow \Init(1^\lambda) \\ (x,z) \leftarrow \aa^{\OH}(pk)}\right.\right] = P_1 + P_2 \end{align}
	\begin{align*}
	\textrm{with } \quad P_1 & \eqdef \Pr_{\hh \leftarrow \RF^{M}_C}\left[\ISCheck(pk,x,\hh(x),z) = 1 \wedge (x \in VC^\IS_{\ge \gamma})  \ \left| \ \substack{(pk,sk) \leftarrow \Init(1^\lambda) \\ (x,z) \leftarrow \aa^{\OH}(pk)}\right.\right]\\
	P_2 & \eqdef \Pr_{\hh \leftarrow \RF^{M}_C}\left[\ISCheck(pk,x,\hh(x),z) = 1 \wedge (x \notin VC^\IS_{\ge \gamma})  \ \left| \ \substack{(pk,sk) \leftarrow \Init(1^\lambda) \\ (x,z) \leftarrow \aa^{\OH}(pk)}\right.\right].
	\end{align*}
	$\aa^{\OH}$ runs in time $t$ so the probability that it outputs $x \in VC^\IS_{\ge \gamma}$ is upper bounded by $QADV_{\IS}^{\rs{\gamma}}(t)$ hence $P_1 \le QADV_{\IS}^{\rs{\gamma}}(t).$ If $x \notin VC^\IS_{\ge \gamma}$ then $|VC_{x}^\IS| \le {\gamma - 1}$. Moreover, if $\ISCheck(pk,x,\hh(x),z) = 1$ then $\hh(x) \in VC^{\IS}_x$. Hence:
	$$ P_2 \le  \Pr_{\hh \leftarrow \RF^{M^{n}}_C}\left[\hh(x) \in VC^{\IS}_x \wedge \left(|VC_x^{\IS}| \le (\gamma-1)\right) \ \left| \ \substack{(pk,sk) \leftarrow \Init(1^\lambda) \\ (x,z) \leftarrow \aa^{\OH}(pk)}\right.\right]. $$
	
	We can directly use Lemma \ref{Lemma:Search} with $U_x = |VC^\IS_x|$ and the fact that $\aa^{\OH}$ perform $q_\hh$ queries to $\OH$ to obtain $P_2 \le O(\frac{q_\hh^2 (\gamma-1)}{|C|}) = O(\frac{q_\hh^2 \gamma}{|C|}).$ Putting the bounds on $P_1$ and $P_2$ in Equation \ref{Eq:S06E01}, we obtain the desired result. 
\end{proof}
This proposition can be seen as a generalization of Unruh's relation between the Fiat-Shamir security and a statistical soundness advantage, but we replace this statistical soundness with rigid soundness. While some schemes may naturally have the rigid soundness property, it is not a priori clear how to use Proposition \ref{Proposition:Generic}. As we will see, this proposition will be very useful when studying {\cao} identification schemes, which we now define and discuss. 

\subsection{Commit and open identification schemes}\label{section:3.5}
A {\cao} identification scheme  is a specific kind of identification scheme where, for the first message, $P$ commits to some values $z_1,\dots,z_n$ using some function $G$ and after the verifier's challenge, he reveals a subset of those values. More precisely, a {\cao} identification scheme $\IS = (\ISKeygen,\ISProver,\ISCheck,G;M,C,R,n)$ consists of the following

\begin{itemize}
	\setlength\itemsep{-0.2em}
	\item A key generation algorithm $\ISKeygen(1^\lambda) \rightarrow (pk,sk)$.
	\item A function $G : R \rightarrow M$ that will act as a commitment scheme.
	\item The challenge set $C$ where each $c \in C$ has a corresponding set $I_c\subseteq [n].$
\item The prover's algorithm $\ISProver = (P_1,P_2)$ for constructing his messages. We have $P_1(sk) \rightarrow (x,z)$ where $z = (z_1,\dots,z_n)$ with each $z_i \in R$ and $x = x_1,\dots,x_n = G(z_1),\dots,G(z_n)$ with each $x_i \in M.$ $P_2(z,c)$ outputs $z_{I_c} = \{z_i\}_{i \in I_c}.$
\item A verification function $\ISCheck(pk,c,z_{I_c})$. The verifier also checks that the commitments are valid, $\ie$ for each $i \in I_c$, $G(z_i) = x_i.$
\end{itemize}
Notice that we now denote by $M$ the message space of individual commited values, so the Prover sends actually an element in $M^n.$ Notice also that in the above verification function, we require $\ISCheck$ to be independent of $x$, and we check the validity of the commitment separately.  All the real  identification schemes we will consider have this property. \\ \\  
\cadre{\begin{center}
		{\Cao} Identification scheme $\IS = (\ISKeygen,\ISProver,\ISCheck,G;M,C,R,n)$
	\end{center}
	\textbf{Initialization.} $(pk,sk) \leftarrow \Keygen(1^\lambda)$. The prover has $(pk,sk)$ and the verifier $pk$. \\
	\textbf{Interaction.}
	\begin{enumerate}
		\setlength\itemsep{-0.2em}
		\item $P$ generates $(z_1,\dots,z_n,G(z_1),\dots,G(z_n)) \leftarrow P_1(sk)$ and sends $x_1,\dots,x_n = G(z_1),\dots,G(z_n)$ to the verifier.
		\item The verifier sends a random $c \Unif C$ that corresponds to a subset $I_c \subseteq [n].$
		\item $P$ sends $z_{I_c}$ to the verifier.
	\end{enumerate}
	\textbf{Verification.} The verifier accepts iff. 
	$ \left(\forall i \in I_c, G(z_i) = x_i\right) \wedge \ISCheck(pk,c,z_{I_c}) = 1.$  \\
}

\paragraph{The Quantum Random Oracle Model for {\cao} identification schemes.}


We will use again the QROM for the commitment function, and model the function $G$ as a random function in  $\RF^{R}_M.$ We will write $\IS_G = (\ISKeygen,\ISProver,\ISCheck,G;M,C,R,n)$ to specify the commitment function used in the subscript of $\IS$. The quantum  Fiat-Shamir advantage therefore becomes
$$ \E_{\substack{\hh \Unif \RF^{M^n}_C \\ G \Unif \RF^{R}_M}} \left[QADV_{\FS{\IS_G}}(t,q_\hh,q_G)\right],$$
where $q_G$ is the number of queries to the unitary $U_G.$

For {\cao} identification schemes, we define $2$ variants of $\gamma$-special soundness. These variants have the nice property that they are independent of the commitment function used, which is not the case for special soundness.  We first define output special soundness

\begin{definition}
	Let $\IS_G = (\ISKeygen,\ISProver,\ISCheck,G;M,C,R,n)$ be a {\cao} identification scheme. For any quantum query algorithm $\aa$, we define 
	\begin{align*} QADV_{\IS_G}^{\osp{\gamma}}(\aa) \eqdef \Pr[|\{c : \ISCheck(pk,c,z_{I_{c}}) = 1\}| \ge \gamma  :  (pk,sk) \leftarrow \Init(1^\lambda), z \leftarrow \aa(pk) 
	\Big]. \end{align*}
	where $z = (z_1,\dots,z_n).$ We also define $QADV_{\IS_G}^{\osp{\gamma}}(t) \eqdef \max_{\aa : |\aa| = t} \left(QADV_{\IS_G}^{\osp{\gamma}}(\aa)\right)$
\end{definition}
The idea of output special soundness is that we can generate $z$ (and $x = G(z_1,\dots,z_n)$ such that there exist $\gamma$ valid triplets $(x,c_1,z_{I_1}),\dots,(x,c_\gamma,z_{I_\gamma})$ for pairwise distinct challenges $c_1,\dots,c_\gamma$. However, the adversary here doesn't need to output these challenges. This notion is incomparable with $\gamma$-special soundness.

The second notion is the $\gamma$-special+ soundness which is the same as above but the adversary has to output the associated challenges. 

\begin{definition}
	Let $\IS_G = (\ISKeygen,\ISProver,\ISCheck,G;M,C,R,n)$ be a {\cao} identification scheme. For any quantum query algorithm $\aa$, we define 
	\begin{multline*} QADV_{\IS_G}^{\spp{\gamma}}(\aa) \eqdef \Pr\Big[\left(\forall i \in [\gamma],  \ISCheck(pk,c_i,z_{I_{c_i}}) = 1\right) \wedge \textrm{ the } c_i \textrm{ are pairwise distinct } : \\  (pk,sk) \leftarrow \Init(1^\lambda), (z,c_1,\dots,c_\gamma) \leftarrow \aa(pk) 
	\Big]. \end{multline*}
	where $z = (z_1,\dots,z_n).$ We also define $QADV_{\IS_G}^{\spp{\gamma}}(t) \eqdef \max_{\aa : |\aa| = t} \left(QADV_{\IS_G}^{\spp{\gamma}}(\aa)\right)$
\end{definition}
This definition is also independent of the commitment used in $\IS$. As the name suggests, this notion is stronger than $\gamma$-special soundness in the sense that $QADV_{\IS_G}^{\spp{\gamma}}(t) \le QADV_{\IS_G}^{\sp{\gamma}}(t).$ This comes from the from an adversary a generating $(z,c_1,\dots,c_\gamma)$ that breaks the $\gamma$-special+ soundness property, we can construct explicitly  $\gamma$ valid triplets $(x,c_1,z_{I_1}),\dots,(x,c_\gamma,z_{I_\gamma})$ with $x = (G(z_1),\dots,G(z_n))$ and the challenges are pairwise distinct, which breaks the $\gamma$-special soundness property. 

We are now ready to jump in the proofs of our theorems. 
\section{The quantum Fiat-Shamir security of {\cao } identification schemes}

\subsection{Overview of our theorems and proof strategy}
Our main theorems are the following:

\begin{theorem}\label{Theorem:ParallelRepetition}
	Let $\IS_G = (\ISKeygen,\ISProver,\ISCheck,G;M,C,R,n)$ be a {\cao} identification scheme where $G$ is modeled as random oracle. Let also $\gamma \ge 2$ be an integer. Using Assumption \ref{Assumption:QROM2}, we have for any $t,q_\hh,q_G:$ 
	$$\E_{\substack{\hh \Unif \RF^{M^n}_C \\ G \Unif \RF^{R}_M}} \left[QADV_{\FS{\IS^{\otimes r}_G}}(t,q_\hh,q_G)\right] \le QADV_{\IS_G}^{\spp{\gamma}}\left(t',q_\hh\right) + O\left(\frac{q_\hh^2 (\gamma-1)^r}{|C|^r}\right) + O_n\left(\frac{(q_G+q_\hh)^3}{|M|}\right).$$
	with $t' = O_n(t) + nr + n|C|.$
\end{theorem}
One can then use $QADV_{\IS_G}^{\spp{\gamma}}\left(t',q_\hh\right) \le QADV_{\IS_G}^{\sp{\gamma}}\left(t',q_\hh\right)$ in order to get a bound in terms of $\gamma$-special soundness.
\begin{theorem}\label{Theorem:Generic}
	Let $\IS_G = (\ISKeygen,\ISProver,\ISCheck,G;M,C,R,n)$ be a {\cao} identification scheme with $G \Unif \RF^R_M$. Let also $\gamma \ge 2$ be an integer. Using Assumption \ref{Assumption:QROM2}, we have for any $t,q_\hh,q_G:$
	\begin{align*} \E_{\substack{\hh \Unif \RF^{M^n}_C \\ G \Unif \RF^{R}_M}} \left[QADV_{\FS{\IS_G}}(t,q_\hh,q_G)\right] \le QADV_{\IS_G}^{\osp{\gamma}}\left(O_n(t),q_\hh\right) + O\left(\frac{q_\hh^2 \gamma}{|C|}\right) +  O_n\left(\frac{(q_G + q_\hh)^3}{|M|}\right).
	\end{align*}
\end{theorem}

\paragraph{Proof strategy.} We present here informally our proof strategy. We fix a {\cao} identification scheme $\IS_G = (\ISKeygen,\ISProver,\ISCheck,G;M,C,R,n)$ and a quantum algorithm $\aa$ that wants to break the quantum soundness of $\FS{\IS_G}.$ This algorithm outputs $x = (x_1,\dots,x_n)$  and $z_{I_c}$ such that if we define $c \eqdef \hh(x)$, we have $\ISCheck(pk,c,z_{I_c})  = 1 \wedge \forall i \in I_c, G(z_i) = x_i.$ If $G$ were an easily invertible permutation, we could from $x = (x_1,\dots,x_n)$ extract the full string $z = (G^{-1}(x_1),\dots,G^{-1}(x_n))$. With such a construction, we can fairly directly relate $\gamma$-rigid soundness and $\gamma$-output special soundness and then conclude using Proposition \ref{Proposition:Generic}. However, $G$ is not usually an efficiently invertible random permutation and it can't be if $\IS$ has to be honest verifier zero-knowledge. In order to circumvent this issue, we perform the $4$ following steps:
\begin{enumerate} \setlength\itemsep{-0.2em}
	\item We transform $\IS$ into $\wIS$ in order to artificially increase the size of $R$. This will allow us to work with larger functions with which we will be able to construct pseudorandom permutations using Feistel networks. 
	\item We start from $G \Unif \RF^R_M$ as our commitment and show that we can replace $G$ with a random permutation $\sigma \in \RP^R.$ 
	\item We now have a random permutation $\sigma$ as our commitment. We show here how to replace $\sigma$ with a random element from a quantum pseudorandom permutation family $\{\wpi_K\}_K$ that is easily invertible. We construct this family using Feistel networks. Of course, we don't mean here that the full identification scheme is secure with this transformation (it is not because it isn't zero-knowledge) but we show the soundness property remains with this transformation. 
	\item Now, that we have an easily invertible permutation, we relate the quantum Fiat-Shamir advantage to the special+ (or output special) soundness advantage of $\wIS$. We can then go back to $\IS$ since the two soundness advantage notions we consider are independent of the commitment used and are the same for $\IS$ and $\wIS$, which allows us to finish the proof. It is only step $4$ that differes for Theorems $\ref{Theorem:ParallelRepetition}$ and $\ref{Theorem:Generic}.$  
\end{enumerate}
We now present these $4$ steps in the next $4$ subsections.

\subsection{Step 1: Transforming $\IS$ into $\wIS$}\label{Section:Step0}
We start from a {\cao} identification $\IS_G = (\ISKeygen,\ISProver,\ISCheck,G;M,C,R,n)$. We consider the smallest set $R'$ of the form $\zo^{2m}$ with $m \ge 2048$ such that $R \subseteq R'$ and $M \subseteq R'$\footnote{To do this, we increase $m$ so that $|M|,|R| \le 2^{2m}$. If this doesn't give us the inclusions then we can relabel the elements of $M$ and $R$ so that they are included in $\zo^{2m}.$}. With this artificial increase of $R$, we consider a commitment function $G' : R' \rightarrow M$. The idea is that instead of committing to each $z_i \in R$ using the string $G(z_i)$, we commit to these strings via the string $G'(z_i||0\dots0),$ where $z_i || 0\dots 0 \in R'.$

We consider $\wIS_{G'} = (\ISKeygen,\ISProver,\widetilde{\ISCheck},G';M,C,R',n)$ that is derived from $\IS_G$ where we changed the space $R$ into $R'$ (and accordingly the function $G$ into $G'$), as well as $\widetilde{\ISCheck}$ which is defined as follows:
$$
\widetilde{\ISCheck}(pk,c,z'_{I_c}) = 1 \Leftrightarrow \left(\forall i \in I_c, \ z'_i = z_i || 0\dots 0 \textrm{ for some } z_i \in R \right) \wedge \ISCheck(pk,c,z_{I_c}) = 1.  
$$
We prove the following proposition
\begin{proposition}\label{Proposition:Step0}
	For any hash function $\hh$, for any $t,q_G,q_\hh$, we have
	$$
	\E_{G \Unif \RF^{R}_M} \left[QADV_{\FS{\IS_G}}(t,q_\hh,q_G)\right] \le 
	\E_{G' \Unif \RF^{R'}_M} \left[QADV_{\FS{\wIS_{{G'}}}}(t,q_\hh,q_G)\right].
	$$
\end{proposition}
\begin{proof}
	The proof is fairly simple and we leave it for Appendix \ref{Appendix:Step0}
\end{proof}

\subsection{Step 2: Replacing $G$ with a random permutation}
We prove the second step, which corresponds to the following proposition.
\begin{proposition}\label{Proposition:Step1} Let $\IS = (\ISKeygen,\ISProver,\ISCheck,G;M,C,R,n)$ be a {\cao} identification scheme with $M \subseteq R.$ We have for any $t,q_\hh,q_G$
	$$ \E_{\substack{\hh \Unif \RF^{M^n}_C \\ G \Unif \RF^{R}_M}} \left[QADV_{\FS{\IS_G}}(t,q_\hh,q_G)\right] \le \E_{\substack{\hh \Unif \RF^{M^n}_C \\ \sigma \Unif \RP^{R}}} \left[QADV_{\FS{\IS_{\sigma}}}(t'',q''_\hh,q''_G) + O_n\left(\frac{(q_G + q_\hh)^3}{|M|}\right)\right].$$
	with $t'' = O_n(t), q''_\hh = q_\hh, q''_G = O_n(\max\{q_G,q_\hh\}).$ 
\end{proposition} 
\begin{proof}
	We first show the following lemma, which states that we can replace $G$ with $\pi \circ G$ for any permutation $\pi \in \RP^{R}$.  In order to define $\pi \circ G$, we actually need to extend $G$ to a function with image $R$, which is possible since we considered the case where $M \subseteq R.$
	\begin{lemma}\label{Lemma:AddPermutation}
	 For any permutation $\pi \in \RP^R$, for which we have an efficient black box access, for any fixed $G \in \RF^R_M$ (extended to $G \in \RF^R_R$), there exists a quantum query algorithm $\bb^{\OH,\OG,\OPI}$ of size $|\bb^{\OH,\OG,\OPI}| \eqdef (t',q'_\hh,q'_G,q'_\pi)$ such that 
		$$\E_{\hh \Unif \RF^{M^n}_C }\left[QADV_{\FS{\IS_G}}(t,q_\hh,q_G)\right] = \E_{\hh \Unif \RF^{M^n}_C}\left[QADV_{\FS{\IS_{\pi \circ G}}}(\bb^{\OH,\OG,\OPI})\right].$$
		and $t' = O_n(t), q'_\hh = q_\hh, q'_G = q_G, q'_\pi = O_n(q_\hh).$
	\end{lemma}
	\begin{proof}
		Let $\aa^{\OH,\OG}$ be a quantum query algorithm with $|\aa^{\OH,\OG}| = (t,q_\hh,q_G)$ and $QADV_{\FS{\IS_G}}(\aa^{\OH,\OG}) = QADV_{\FS{\IS_G}}(t,q_\hh,q_G).$  Fix also a permutation $\pi$. For each function $\hh : M^n \rightarrow C$, we define 
		$\hh_{\pi}(x_1,\dots,x_n) \eqdef \hh(\pi(x_1),\dots,\pi(x_n)).$ Notice that if $\hh \Unif \RF^{M^n}_C$ then $\hh_\pi$ is also uniformly random  in $\RF^{M^n}_C$ for any fixed $\pi$. Therefore, we have 
		\begin{align}\label{Eq:S05E01}
		\E_{\hh \Unif \RF^{M^n}_C }\left[QADV_{\FS{\IS_G}}(\aa^{\OH,\OG})\right] = \E_{\hh \Unif \RF^{M^n}_C }\left[QADV_{\textup{FS}^{\hh_\pi}[{\IS_{G}}]}(\aa^{\ket{\hh_\pi},\OG})\right].\end{align}
		We now construct the following algorithm $\bb^{\OH,\OG,\OPI}: (x,z_{I_c}) \leftarrow \aa^{\ket{\hh_\pi},\OG}, \textrm{ return } (\vpi(x),z_{I_c}).$
		where we use the notation ${\vpi}(x) = \pi(x_1),\dots,\pi(x_n).$ The algorithm 
		$\bb^{\OH,\OG,\OPI}$ emulates calls to $U_{\hh_\pi}$, with calls to $U_\hh$ and $U_\pi$, using each time $n$ calls to $U_\pi$ and $1$ call to $U_\hh$. 
		\begin{align}
		\Gamma_1 & \eqdef QADV_{\textup{FS}^{\hh_\pi}[{\IS_{G}}]}(\aa^{\ket{\hh_\pi},\OG}) \nonumber  =  \Pr\left[\ISCheck(pk,c,z_{I_c}) = 1 \wedge \left(\forall i \in I_c, G(z_i) = x_i\right)  \ \left| \ \substack{(pk,sk) \leftarrow \Init(1^\lambda) \\ (x,z_{I_c}) \leftarrow \aa^{\ket{\hh_\pi},\OG} \\ c = \hh_\pi(x)}\right.\right] \nonumber \\
		& = \Pr\left[\ISCheck(pk,c,z_{I_c}) = 1 \wedge \left(\forall i \in I_c, (\pi \circ G)(z_i) = \pi(x_i)\right)  \ \left| \ \substack{(pk,sk) \leftarrow \Init(1^\lambda) \\ (x,z_{I_c}) \leftarrow \aa^{\ket{\hh_\pi},\OG} \\ c = \hh(\vpi(x))}\right.\right] \nonumber \\
		& = \Pr\left[\ISCheck(pk,c,z_{I_c}) = 1 \wedge \left(\forall i \in I_c, (\pi \circ G)(z_i) = \pi(x_i)\right)  \ \left| \ \substack{(pk,sk) \leftarrow \Init(1^\lambda) \\ (\vpi(x),z_{I_c}) \leftarrow \bb^{\OH,\OG,\OPI} \\ c = \hh(\vpi(x))}\right.\right] \nonumber \\
		& = QADV_{\FS{\IS_{\pi \circ G}}}(\bb^{\OH,\OG,\OPI}). \label{Eq:S05E02}
		\end{align}
		Combining Equations \ref{Eq:S05E01} and \ref{Eq:S05E02}, we can conclude
		\begin{align*}
		\E_{\hh \Unif \RF^{M^n}_C }\left[QADV_{\FS{\IS_G}}(\aa^{\OH,\OG})\right] & = \E_{\hh \Unif \RF^{M^n}_C }\left[QADV_{\textup{FS}^{\hh_\pi}[{\IS_{G}}]}(\aa^{\ket{\hh_\pi},\OG})\right] \\
		& =  \E_{\hh \Unif \RF^{M^n}_C} \left[QADV_{\FS{\IS_{\pi \circ G}}}(\bb^{\OH,\OG,\OPI})\right].
		\end{align*}	\qed
	\end{proof}
	We now go back to the proof of Proposition \ref{Proposition:Step1}. The above lemma holds for any $\pi$ and $G$, so we can choose in particular a random function $G$ and random permutation $\pi$, which gives us 
	\begin{align}\label{Equation:3}
	\E_{\substack{\hh \Unif \RF^{M^n}_C \\ G \Unif \RF^{R}_M}} \left[QADV_{\FS{\IS_G}}(t,q_\hh,q_G)\right] & =  \E_{\hh \Unif \RF^{M^n}_C}\E_{\substack{G \Unif \RF^{R}_M \\ \pi \Unif \RP^R}} \left[QADV_{\FS{\IS_{\pi \circ G}}}(\bb^{\OH,\OG,\OPI})\right] \\
	& \le \E_{\substack{\hh \Unif \RF^{M^n}_C \\ \sigma \Unif \RP^R}}\left[QADV_{\FS{\IS_{\sigma}}}(t,q'_\hh,q'_\sigma)\right] +  O\left(\frac{(q'_\sigma)^3}{|M|}\right).
	\end{align}
with $q''_G = \max\{q'_G,q'_\pi\} = O_n(\max\{q_G,q_\hh\})$, where the last inequality comes from Assumption \ref{Assumption:QROM2}.
\end{proof}
\subsection{Step 3: Replacing the random permutation $\sigma$ with an efficiently invertible QPRP}
We assume there exists a family of quantum secure pseudorandom functions $\{f_K\}$ where each $f_K : R \rightarrow R$. We can use for example $f_K = \KMAC_K$. We define 
$\wpi_K \eqdef \FeF(f_K)$. 
\begin{proposition}\label{Proposition:Step2}
	Let $\IS = (\ISKeygen,\ISProver,\ISCheck,G;M,C,R,n)$ be a {\cao} identification scheme with $R = \zo^{2m}$ for some integer $m$. For any fixed $\hh$: 
	$$ \E_{\sigma \Unif \RP^R}\left[QADV_{\FS{\IS_{\sigma}}}(t,q_\hh,q_G)\right] \le \E_{K} \left[QADV_{\FS{\IS_{\widetilde{\pi}_K}}}(t,q_\hh)\right] + O({\frac{q_G^3}{2^{m/2}}}).$$
\end{proposition}	
\begin{proof}
	Now fix $\hh$. We have 
	\begin{align}	
	\E_{\sigma \Unif \RP^R}\left[QADV_{\FS{\IS_{\sigma}}}(t,q_\hh,q_G)\right]&  \le 
	\E_{f \Unif \RF^{\zo^m}_{\zo^m}}\left[QADV_{\FS{\IS_{\FeF(f)}}}(t,q_\hh,q_G)\right] + O(\sqrt{\frac{q_G^6}{2^{m}}}) \\
	& = \E_{K \Unif \mathcal{K}} \left[QADV_{\FS{\IS_{\widetilde{\pi}_K}}}(t,q_\hh)\right] + O({\frac{q_G^3}{2^{m/2}}}).
	\end{align}
	The first inequality comes from Proposition \ref{Proposition:Feistel} and the second equality comes from our assumption that $f_K$ is a pseudorandom family. 
\end{proof}

When $m \ge 2048$ (this is the value chosen in Step $1$ but it could have been another arbitrary large value), the term $O(\frac{q_G^3}{2^{m/2}})$ will always be tiny and irrelevant for the amounts of security we consider. 

\subsection{Finishing the proof: step 4 and conclusion}
So we managed to replace the commitment function by a permutation $\widetilde{\pi}_K = \FeF(f_K)$ for a randomly chosen $K$. As we described in Section \ref{Section:Preliminaries}, the use of Feistel networks for constructing $\wpi_K$ implies that both $\widetilde{\pi}_K$ and $\widetilde{\pi}^{-1}_K$ are efficiently computable without needing to know how to compute preimages for $f_K$ Our goal in this final step is to bound $\E_K\left[QADV_{\FS{\IS_{\widetilde{\pi}_K}}}(t,q_\hh)\right].$
\subsubsection{Step 4 used for Theorem \ref{Theorem:Generic}}
We actually reason here for a fixed key $K$ we have $\wpi_K = \wpi$.

\begin{proposition}\label{Proposition:Step3}
	Let $\IS_{\widetilde{\pi}} = (\ISKeygen,\ISProver,\ISCheck,\widetilde{\pi};M,C,R,n)$ be a {\cao} identification scheme where $\wpi$ is efficiently computable and invertible. For any integer $\gamma \ge 2,$ for any fixed $\hh$, we have
	$$ QADV_{\FS{\IS_{\widetilde{\pi}}}}(t,q_\hh) \le QADV_{\IS_{\widetilde{\pi}}}^{\osp{\gamma}}(t + n,q_\hh) + O(\frac{q^2_\hh \gamma}{|C|}). $$
	Notice here that since $\widetilde{\pi}$ has a known efficient description, we don't consider only black box calls to $U_{\widetilde{\pi}}$ but we can perform any computation that depends on the description of $\widetilde{\pi}$ and $\widetilde{\pi}^{-1}.$
\end{proposition}\

\begin{proof}
	Fix a {\cao} identification scheme $\IS_{\widetilde{\pi}} = (\ISKeygen,\ISProver,\ISCheck,\widetilde{\pi};M,C,R,n)$, and an integer $\gamma \ge 2$. Using Proposition \ref{Proposition:Generic}, we have 
	
	\begin{align}\label{Eq:Partial}
	QADV_{\FS{\IS_{\widetilde{\pi}}}}(t,q_\hh) \le QADV_{\IS_{\widetilde{\pi}}}^{\rs{\gamma}}(t,q_\hh) + O(\frac{q^2_\hh \gamma}{|C|}).\end{align}
	
	 Let $\cc^{\OH}$ be an quantum query algorithm satisfying $|\cc^{\ket{\hh}}| = (t,q_\hh)$ and  $QADV_{\IS_{\widetilde{\pi}}}^{\rs{\gamma}}(t,q_\hh) = QADV_{\IS_{\widetilde{\pi}}}^{\rs{\gamma}}(\cc^{\OH}).$ We consider the following algorithm $\bb^{\OH}$:
	$$\bb^{\OH}(pk) : x \eqdef (x_1,\dots,x_n) \leftarrow \cc^{\OH}(pk), z = ({\widetilde{\pi}}^{-1}(x_1),\dots,{\widetilde{\pi}}^{-1}(x_n)), \textrm{ return } z. $$ 
	
	Notice that if $\cc^{\OH}$ outputs a value $x \in VC^\IS_{\ge \gamma}$, then $|\{c : \ISCheck(pk,c,z_{I_c}) = 1\}| \ge \gamma).$ Therefore, $QADV_{\IS_{\widetilde{\pi}}}^{\rs{\gamma}}(\cc^{\OH}) \le QADV_{\IS_{\widetilde{\pi}}}^{\osp{\gamma}}(\bb^{\OH}).$ Also $\bb^{\OH}$ runs in time $t + n$ (recall that ${\widetilde{\pi}}^{-1}$ can be performed efficiently so we consider here its running time is $1$). We can therefore conclude
	$$ QADV_{\IS_{\widetilde{\pi}}}^{\rs{\gamma}}(t,q_\hh) \le QADV_{\IS_{\widetilde{\pi}}}^{\gamma\textrm{-}osp}(t + n ,q_\hh). $$
\end{proof}

\subsubsection{Theorem \ref{Theorem:Generic}: putting everything together} \label{Section:redo}

We can now show our first main theorem, which is the combination of our $4$ steps.
\setcounter{theorem}{1}
\begin{theorem}
	Let $\IS_G = (\ISKeygen,\ISProver,\ISCheck,G;M,C,R,n)$ be a {\cao} identification scheme with $G \Unif \RF^R_M$. Let also $\gamma \ge 2$ be an integer. We have for any $t,q_\hh,q_G:$
	\begin{align*} \E_{\substack{\hh \Unif \RF^{M^n}_C \\ G \Unif \RF^{R}_M}} \left[QADV_{\FS{\IS_G}}(t,q_\hh,q_G)\right] \le QADV_{\IS_G}^{\osp{\gamma}}\left(O_n(t),q_\hh\right) + O\left(\frac{q_\hh^2 \gamma}{|C|}\right) +  O_n\left(\frac{(q_G + q_\hh)^3}{|M|}\right).
	\end{align*}
\end{theorem}
\begin{proof}
	We start from $\IS_G$ and construct $\wIS_{G'} = (\ISKeygen,\ISProver,\widetilde{\ISCheck},G';M,C,R',n)$ as in Proposition \ref{Proposition:Step0}. We have in particular $M \subseteq R$, which allows us to apply Proposition \ref{Proposition:Step1} and $R = \zo^{2m}$ with $m \ge 2048.$ We define  $\widetilde{\pi}_K \eqdef \FeF(f_K)$ where $\{f_K\}$ is a quantum pseudorandom function family and $\mathcal{K}$ is the key space. We have
	\begin{align*}
	\Gamma_2 & \eqdef \E_{\substack{\hh \Unif \RF^{M^n}_C \\ G \Unif \RF^{R}_M}} \left[QADV_{\FS{\IS_G}}(t,q_\hh,q_G)\right]  \le 	\E_{\substack{\hh \Unif \RF^{M^n}_C \\ G' \Unif \RF^{R'}_M}} \left[QADV_{\FS{\wIS_{{G'}}}}(t,q_G,q_\hh)\right] \\
& \le \E_{\substack{\hh \Unif \RF^{M^n}_C \\ \sigma \Unif \RP^{R}}} \left[QADV_{\FS{\wIS_{\sigma}}}(O_n(t),q_\hh,O_n(q_G + q_\hh))\right]  + O_n\left(\frac{(q_G + q_\hh)^3}{|M|}\right)\\
	& =  \E_{\substack{\hh \Unif \RF^{M^n}_C \\ K \Unif \mathcal{K}}} \left[QADV_{\FS{\wIS_{\widetilde{\pi}_K}}}(O_n(t),q_\hh)\right]  + O_n\left(\frac{(q_G + q_\hh)^3}{|M|}\right)\\
	& \le \E_{K \Unif \mathcal{K}} \left[QADV_{\wIS_{\wpi_K}}^{\osp{\gamma}}\left(O_n(t),q_\hh\right) \right]+ O\left(\frac{q_\hh^2 \gamma}{|C|}\right) + O_n\left(\frac{(q_G + q_\hh)^3}{|M|}\right) \\
	& = QADV_{\IS_{G}}^{\osp{\gamma}}\left(O_n(t),q_\hh\right) + O\left(\frac{q_\hh^2 \gamma}{|C|}\right) + O_n\left(\frac{(q_G + q_\hh)^3}{|M|}\right)
	\end{align*}
	The first $4$ lines come from the $4$ steps of our proof, namely Propositions \ref{Proposition:Step0}, \ref{Proposition:Step1}, \ref{Proposition:Step2} and \ref{Proposition:Step3}. We ignored the term $O_n(\frac{(q_G + q_\hh)^3}{2^{2m}})$ from Proposition \ref{Proposition:Step2} which is tiny and absorbed by the other terms for any reasonable security requirement since $m \ge 2024.$ For the last inequality, we remove the dependency in $\wpi_K$ because the quantity $QADV_{\IS}^{\osp{\gamma}}$ is independent of the commitment used, and we can go back from $\widetilde{\IS}$ to $\IS$ by noticing that the $\gamma$-output special soundness is the same for these $2$ identification schemes.  This concludes the proof of Theorem \ref{Theorem:Generic}.
\end{proof}

\subsubsection{Step 4 used for Theorem \ref{Theorem:ParallelRepetition}}
We prove here the Step $4$ that will be used for proving Theorem \ref{Theorem:ParallelRepetition}.

\begin{proposition}\label{Proposition:Step3Parallel}
	Let $\IS_{\widetilde{\pi}} = (\ISKeygen,\ISProver,\ISCheck,\widetilde{\pi};M,C,R,n)$ be a {\cao} identification scheme, where $\wpi$  is an efficiently computable and invertible permutation. For any integer $\gamma \ge 2$, and $r \in \mathbb{N}^*$, we have
	\begin{align*} QADV_{\FS{\IS^{\otimes r}_{\wpi}}}(t,q_\hh) \le  QADV_{\IS_{\wpi}}^{\spp{\gamma}}(t + nr + |C|r,q_\hh) + O\left(\frac{q_\hh^2 (\gamma - 1)^r}{|C|^r}\right). \end{align*}
\end{proposition}
\begin{proof}
	
 We first show the following $2$ lemmata.

\begin{lemma}
	Let $S \subseteq |C|^r$. Let $\gamma \ge 2$ be an integer. If $|S| \ge (\gamma - 1)^r + 1$ then there exists an index $i \in [r]$, $|\{c_i : \exists c = (c_1,\dots,c_{r}), c \in S\}| \ge \gamma.$
\end{lemma}
\begin{proof}
	We  prove the contrapositive. Let $T_i \eqdef \{c_i : \exists c = (c_1,\dots,c_{r}), c \in S\}$ and assume that $\forall i \in [r], \ |T_i| \le \gamma - 1$. We immediately have 
	$S \subseteq T_1 \times T_2 \dots \times T_r$ which implies $|S| \le \Pi_{i \in {r}} |T_i| \le (\gamma - 1)^r.$
\end{proof}

For the next lemma, recall the definitions of valid challenges of Section \ref{Section:RigidSoundness}. We wil now write $\IS^{\otimes r}$ instead of $\IS^{\otimes r}_{\wpi}$ to lighten the notations.

\begin{lemma}  \label{Lemma:3G}
	Let $x = (x^1,\dots,x^r) \in M^{nr}$ where for each $i \in [r], x^i = (x^i_1,\dots,x^i_n)$ and each $x^i_j \in M.$ Let also $\gamma \ge 2$ be an integer. If $x \in V^{\IS^{\otimes r}}_{\ge (\gamma - 1)^r + 1}$ then there exists an $i \in [r]$, distinct values $b_1,\dots,b_{\gamma} \in C$ such that if we define $z^i \eqdef {\widetilde{\vpi}_0}^{-1}(x^i) = (\wpi^{-1}(x^i_1),\dots,\wpi^{-1}(x^i_n))$, we have $\forall j \in [\gamma], \ \ISCheck(pk,b_j,z^i_{I_{b_j}}) = 1$.
\end{lemma}
\begin{proof}
	Fix $x \in M^{nr}$ and assume $|VC^{\IS^{\otimes r}}_x| \ge (\gamma - 1)^r + 1.$  Using the previous lemma, let $i \in [r]$ be the index such that $|\{c_i \in C : \exists c = (c_1,\dots,c_{r}) \in VC^{\IS^{\otimes r}}_x\}| \ge \gamma$ and we denote by $\{b_1,\dots,b_{\gamma}\}$ any $\gamma$ pairwise distinct values of this set. For each $j \in [\gamma]$, let $c^j \in VC^{\IS^{\otimes r}}_x$ such that $c^j_i = b_j$.   Let $z = \widetilde{\vpi}_0^{-1}(x)$. This means for each $i \in [r]$, $z^i = \widetilde{\vpi}_0^{-1}(x^i).$ Now
	$\forall j \in [\gamma],$ because the strings $b_i \in VC^{\IS^{\otimes r}}_x$, we have 
	$$ \forall j \in [\gamma], \ \ISCheck^{\otimes r}(x,c^j,z_{I_{c^j}}) = 1 \quad \Rightarrow \quad \forall j \in [\gamma], \ \ISCheck(pk,b_j,z^i_{I_{b_j}}) = 1.$$ 
\end{proof}

\noindent With these $2$ lemmata, we can prove Proposition \ref{Proposition:Step3Parallel}. First notice using Proposition \ref{Proposition:Generic} that 
\begin{align}\label{Equation:New}
QADV_{\FS{\IS^{\otimes r}_{\wpi}}}(t,q_\hh) \le QADV_{\IS^{\otimes r}_{\wpi}}^{\rs{((\gamma - 1)^r + 1)}}(t,q_\hh) + O\left(q_\hh^2 \frac{(\gamma-1)^r}{|C|^r}\right).
\end{align}

	Let $\aa^{\OH}$ be a quantum algorithm running with $|\aa^{\OH}| = (t,q_\hh)$ such that $QADV_{\IS^{\otimes r}_{\wpi}}^{\rs{((\gamma-1)^r + 1)}}(\aa^{\OH}) = QADV_{\IS^{\otimes r}_{\wpi}}^{\rs{((\gamma-1)^r + 1)}}(t,q_\hh).$ We consider the following algorithm $\bb^{\OH}$: \\ \\
	\cadre{\begin{center} Quantum algorithm $\bb^{\OH}$
		\end{center}
		\begin{enumerate} \setlength\itemsep{-0.2em}
			\item compute $x = x^1,\dots,x^r \leftarrow \aa^{\OH}(p_k)$ where for each $i \in [r], x^i = x^i_1,\dots,x^i_n.$
			\item compute for each $i \in [r], j \in [n] \ z^i_j = \wpi^{-1}(x^i_j).$ Similarly as above, we define $z^i = z^i_1,\dots,z^i_n$ for each $i \in [r].$
			\item Find $i \in [r]$ and distinct values $b_1,\dots,b_{\gamma} \in C$ such that for each $j \in [\gamma], \ISCheck(pk,b_j,z^i_{I_j}) = 1$ if such values exist, else output $\bot.$ To do so, we compute $\ISCheck(pk,b,z^i_{I_b})$ for each $i \in [r]$ and $b \in C$.
			\item Output $(b_1,\dots,b_{\gamma},z^i)$.
		\end{enumerate}
	} $ \ $ \\ \\
	Using Lemma \ref{Lemma:3G}, we have 
	\begin{align}
	QADV_{\IS_{\wpi}^{\otimes r}}^{\rs{((\gamma - 1)^r + 1)}}(\aa^{\OH}) & =  \Pr\left[x \in VC^\IS_{(\gamma - 1)^r + 1}  \ \left| \ \substack{(pk,sk) \leftarrow \Init(1^\lambda) \\ x \leftarrow \aa^{\OH}(pk)}\right.\right] \nonumber \\
	& \le \Pr\big[\exists i \in [r], \exists \textrm{ distinct } b_1,\dots,b_{\gamma} \in C: \nonumber \\ & \qquad \qquad \forall j \in [\gamma],  \ISCheck(pk,x^i,b_j,(\wpi^{-1}(x^i))_{I_{b_j}}) = 1  \ {\huge|} \ \substack{(pk,sk) \leftarrow \Init(1^\lambda) \\ x \leftarrow \aa^{\OH}(pk)}\big] \nonumber \\
	& = \Pr\left[\forall j \in [\gamma], \ISCheck(pk,b_j,z^i_{I_j}) = 1  \ \left| \ \substack{(pk,sk) \leftarrow \Init(1^\lambda) \\ (b_1,\dots,b_{\gamma},z^i) \leftarrow \bb^{\OH}(pk)}\right.\right] \nonumber \\
	& = QADV_{\IS_{\wpi}}^{\spp{\gamma}}(\bb^{\OH}) \label{Equation:New2}
	\end{align}
	Now, let's compute the running time of $\bb^{\OH}$. Step 1: takes time $t$. Step 2: makes $nr$ calls to $\wpi_0^{-1}$, which is efficiently computable. Step $3$: makes $|C|r$ calls to $\ISCheck$ which is efficiently computable. This is implies that the total running time of $\bb$ is $t + nr + |C|r.$ Moreover, $\bb^{\OH}$ makes as much queries to $\OH$ as $\aa^{\OH}.$	Combining Equation \ref{Equation:New} and \ref{Equation:New2}, we conclude 
$$QADV_{\FS{\IS^{\otimes r}_{\wpi}}}(t,q_\hh)  \le QADV_{\IS_{\wpi}}^{\spp{\gamma}}(t + nr + |C|r,q_\hh) + O\left(\frac{q_\hh^2 (\gamma - 1)^r}{|C|^r}\right)$$
\end{proof}

\subsubsection{Finishing the proof of Theorem \ref{Theorem:ParallelRepetition}}
We finish the proof exactly as we did for Theorem \ref{Theorem:Generic} in Section \ref{Section:redo} except we replace Proposition \ref{Proposition:Step3} with Proposition \ref{Proposition:Step3Parallel}.

\section{Practical instantiations}
Assume we have a {\cao} identification scheme $\IS^{\otimes r}$ where $\IS$ has challenge size $3$ and has $3$-special soundness, in the sense that $QADV_{\IS}^{\sps{3}}(t)$ is smaller that the probability to break the underlying hard computational problem in time $t$. From Theorem \ref{Theorem:ParallelRepetition}, we have that $\S_\IS$ has $\lambda$ bits of security by taking $r$ such that $2^{2\lambda}(\frac{2}{3})^r < 1$ or equivalently $r   *\log(2/3) = 2\lambda$ and $\log_2(|M|) \ge 3 \lambda$.

For example, if we take $\lambda = 64$, we have $r \ge 219$ and $|M| \ge 192$. If we take $\lambda = 128$, this gives $r \ge 438$ and $|M| \ge 384$. This kind of bounds applies to Stern's identification scheme, the \cite{KTX08} identification scheme based on lattice problems, the \cite{SSH11} identification schemes based on multivariate  problems, closely related to the NIST candidate MQDSS, and the PICNIC scheme based on multiparty computing problems, which is also a NIST candidate. Actually, for the NIST candidates, this rationale was already used so our results essentially claim that this can be done with a provable tight security reduction in the QROM. The only difference is that the commitment scheme used is a call to $\SHAKE$ with has $512$ output bits (so $|M|$) and a possible improvements of these schemes would be to reduce the size of $|M|$. In order to show our derivations more in detail, we present in Appendix \ref{Appendix:Stern} these derivations for Stern's signature scheme and similar derivation can be done for the other signature schemes mentioned above. 

The $5$ round schemes, such as MQDSS or the KKW variant of PICNIC seem to require more work but the current techniques seem quite promising for proving tight security reductions for those as well. There are also more complicated schemes that are {\cao} but with more rounds such as Pigroast/Legroast \cite{BD20}. These multi-round protocols also have asymptotic quantum reductions from the work of \cite{DFM20} and we hope our techniques can be useful here for concrete security. We leave this for future work. \vspace*{-0.5cm}

.

\bibliography{paper_shrinked}
\bibliographystyle{alpha}

\appendix

\newpage

\begin{center}{\LARGE \textbf{Appendix}} \end{center}
\section{Signature schemes}\label{Appendix:Signatures}
A signature scheme $S$ consists of $3$ algorithms  $(\SKeygen,\SSign,\SVerify)$:
\begin{itemize}
	\item $\SKeygen(1^\lambda) \rightarrow (pk,sk)$ is the generation of the public key $pk$ and the secret key $sk$ from the security parameter $\lambda$.
	\item $\SSign(m,pk,sk) \rightarrow \sigma_m$ : generates the signature $\sigma_m$ of a message $m$ from $m,pk,sk$.
	\item $\SVerify(m,\sigma,pk) \rightarrow \zo$ verifies that $\sigma$ is a valid signature of $m$ using $m,\sigma,pk$. The output $1$ corresponds to a valid signature.
\end{itemize}
\paragraph{Correctness.} A signature scheme is correct iff. when we sample $(pk,sk) \leftarrow \SKeygen(1^\lambda)$, we have for each $m$
$$ \SVerify(m,\SSign(m,pk,sk),pk) = 1.$$
\paragraph{Security definitions}
We consider the standard EUF-CMA security for signature schemes. To define the advantage of an adversary $\aa$, we consider the following interaction with a challenger:
\\ \\
\textbf{Initialize.} The challenger generates $(pk,sk) \leftarrow \SKeygen(1^\lambda)$ and sends $pk$ to $\aa$. \\
\textbf{Query phase.} $\aa$ can perform sign queries by sending each time a message $m$ to the challenger who generates $\sigma = \SSign(m,pk,sk)$ and sends $\sigma$ to $\aa$. Let $m_1,\dots,m_{q_S}$ the (not necessarily distinct) queries made by $\aa$. The adversary can also make $q_\hh$ queries to $\hh$. \\
\textbf{Output.} $\aa$ outputs a pair $(m^*,\sigma^*)$. The advantage $Adv(\aa)$ for $\aa$ is the quantity
\begin{align*}
QADV_{\ss}^{\textup{EUF-CMA}}(\aa)  = \Pr[\aa & \textrm{ outputs } (m^*,\sigma^*) \ st. \\ & \SVerify(m^*,\sigma^*,pk) = 1 \wedge m^* \neq m_1,\dots,m_{q_S}], 
\end{align*}
where $m^* \neq m_1,\dots,m_{q_S}$ means $\forall i, \ m^* \neq m_{i}.$
\begin{definition}\label{Definition:Signatures-IND-CCA}
	Let $\ss = (\SKeygen,\SSign,\SVerify)$ be a signature scheme. We define
	$$ QADV_{\ss}^{\textup{EUF-CMA}}(t,q_\hh,q_S) = \max_{\aa} QADV_{\ss}^{\textup{EUF-CMA}}(\aa).$$
	where we maximize over an adversary running in time $t$, performing $q_\hh$ hash queries and $q_S$ sign queries. 
\end{definition}

We can directly construct a signature scheme from an identification scheme via the Fiat-Shamir transform. From an identification scheme $\IS = (\ISKeygen,\ISProver = (P_1,P_2),\ISCheck;M,C,R)$, we define the following signature scheme \\ $\ss_{\IS} = (\SISKeygen,\SISSign,\SISVerify)$ that uses a random function $\hh$: 
\begin{itemize}
	\item $\SISKeygen(1^\lambda) = \Init(1^\lambda)$
	\item $\SISSign(m,pk,sk) : (x,St) \leftarrow P_1(pk), c \leftarrow \hh(x,m), z \leftarrow P_2(sk,x,c,St)$, output $\sigma = (x,z)$.
	\item $\SISVerify(m,\sigma = (x,z),pk) = V(pk,x,\hh(x,m),z)$.
\end{itemize}

\begin{proposition}\cite{GHHM20}\label{Proposition:KLS18}
	Let $\IS$ be an identification scheme which is $\eps$-HVZK and has $\alpha$ bits of min-entropy. Let $S_{\IS}$ the corresponding signature scheme. 
	$$QADV_{S_{\IS}}^{\textrm{EUF-CMA}}(t,q_\hh,q_S) \le QADV_{\FS{\IS}}(t') + \frac{3q_{\ss}}{2}\sqrt{(q_\ss + q_\hh + 1)2^{-\alpha}} + QADV_{q_{\ss} \eps}.$$
	where we need to average the $2$ advantages over the hash function $\hh.$
\end{proposition}

The min-entropy here is the min-entropy of the prover's first message when he is honest. All schemes we consider can have very large min-entropy using the method presented for instance in\cite{KLS18} with marginal cost, so the  term $\frac{3q_{\ss}}{2}\sqrt{(q_\ss + q_\hh + 1)2^{-\alpha}}$ can be made  small.  The above proposition shows that we only need to focus on the soundness of the Fiat-Shamir transform in order to build signature schemes, which is what we will do in the paper. Notice that such a proposition holds also if we only consider computational zero-knowledge, we refer to \cite{GHHM20} for more details. 

\section{Proof of the first step}\label{Appendix:Step0}
We prove the following proposition
\begin{proposition}[Proposition \ref{Proposition:Step0} restated]
	For any hash function $\hh$, for any $t,q_G,q_\hh$, we have
	$$
	\E_{G \Unif \RF^{R}_M} \left[QADV_{\FS{\IS_G}}(t,q_\hh,q_G)\right] \le 
	\E_{G' \Unif \RF^{R'}_M} \left[QADV_{\FS{\wIS_{{G'}}}}(t,q_\hh,q_G)\right].
	$$
\end{proposition}
	\begin{proof}

		For any function $G' \in \RF^{R'}_M$, we define the function $C_{G'} \in \RF^{R}_M$ as follows: $C_{G'}(z) \eqdef G'(z||0\dots 0).$ Notice that if $G'$ is a random function in $\RF^{R'}_M$ then $C_{G'}$ is a random function in $\RF^R_M$. Therefore
		\begin{align*}
			\E_{G \Unif \RF^{R}_M} \left[QADV_{\FS{\IS_G}}(\aa^{\ket{G},\ket{\hh}})\right] & = 
			\E_{G' \Unif \RF^{R'}_M} \left[QADV_{\FS{\IS_{C_{G'}}}}(\aa^{\ket{C_{G'}},\ket{\hh}})\right]. 
		\end{align*}
		Now, let's consider the following algorithm $\aa_2^{\ket{C_{G'}},\ket{\hh}} : (x,z_{I_c}) \leftarrow \aa^{\ket{C_{G'}},\ket{\hh}} \textrm{ with } c = \hh(x).$ { Return } $(x,z'_{I_c})$ { where } $\forall i \in I_c, \ z'_i = z_i || 0 \dots 0.$ From the definition of $\wIS$, we have 
		\begin{align*}
			\E_{G' \Unif \RF^{R'}_M} \left[QADV_{\FS{\IS_{C_{G'}}}}(\aa^{\ket{C_{G'}},\ket{\hh}})\right] & =
			\E_{G' \Unif \RF^{R'}_M} \left[QADV_{\FS{\wIS_{{G'}}}}(\aa_2^{\ket{C_{G'}},\ket{\hh}})\right] \\ & \le
			\E_{G' \Unif \RF^{R'}_M} \left[QADV_{\FS{\wIS_{{G'}}}}(\aa_2^{\ket{{G'}},\ket{\hh}})\right]
		\end{align*}
		where the last inequality comes from the fact that a call to $U_{C_{G'}}$ can be done with a call to $U_{G'}.$ Since the running time and number of queries remains unchanged between $\aa$ and $\aa_2$, we can conclude.\qed
	\end{proof}
	\noindent Notice also from the definitions that we can derive the following equalities, for any $\gamma$ and $\hh$:
	\begin{align*}
		\E_{G \Unif \RF^{R}_M} \left[QADV^{\spp{\gamma}}_{{\IS_G}}(t,q_G,q_\hh)\right] & =
		\E_{G' \Unif \RF^{R'}_M} \left[QADV^{\spp{\gamma}}_{{\wIS_{{G'}}}}(t,q_G,q_\hh)\right] \\
		\E_{G \Unif \RF^{R}_M} \left[QADV^{\osp{\gamma}}_{{\IS_G}}(t,q_G,q_\hh)\right] & =
		\E_{G' \Unif \RF^{R'}_M} \left[QADV^{\osp{\gamma}}_{{\wIS_{{G'}}}}(t,q_G,q_\hh)\right]
	\end{align*}

\section{Stern signature scheme}\label{Appendix:Stern}
\paragraph{Notations for this section.} Matrices are denoted with bold large letters, for eg. $\Mm$ and line vectors will be denoted with bold small letters, for eg. $\vv = (v_1,\dots,v_n).$ The Hamming weight $|\cdot|_H$ for binary vectors is defined as follows: $|\vv|_H = |\{i : v_i = 1 \}|.$ \\

\noindent Stern's signature scheme is one of the first signature schemes based on a {\cao} identification scheme. It is a post-quantum signature scheme based on the hardness of the syndrome decoding problem, which is the canonical hard problem for code-based cryptography.

\begin{problem}[Syndrome Decoding - \textup{SD}$(n,k,w)$]\label{prob:decoGenR}	
	$ \ $ 
	\begin{itemize}		\setlength\itemsep{-0.2em}
		\item \textup{Instance:} a parity-check matrix $\Hm \in \zo^{(n-k)\times n}$ of rank $n-k$, a syndrome $\sv \in \zo^{n-k}$,
		\item \textup{Output:} $\ev\in S_w$ such that $\ev\transpose{\Hm} = \sv$  where $S_w \eqdef \{\ev \in \zo^n : |\ev|_H = w\}.$
	\end{itemize}
\end{problem}
\noindent We also define the syndrome decoding advantage:

\begin{definition}[SD-advantage$(n,k,w)$]
	For any algorithm $\aa$, we define
	$$
	Adv_{(n,k,w)}^{\textup{SD}}(\aa) \eqdef \Pr\left(\ev\transpose{\Hm} = \sv \wedge |\ev| = w \left| \Hm \Unif \FR^{(n-k),n}, \  \sv \Unif \zo^{n-k}, \ \ev \leftarrow \aa(\Hm,\sv) \right.\right), 
	$$
	where $\FR^{(n-k),n} \eqdef \{\Hm \in \zo^{(n-k)\times n} : \Hm \textrm{ has rank } (n-k)\}$ is the set of full rank matrices in $\zo^{(n-k),n}$. For any time $t$, we also define,
	$
	Adv^{\textup{SD}}_{(n,k,w)}(t)  \eqdef \max_{\aa : |\aa| = t} Adv^{\textup{SD}}_{(n,k,w)}(\aa).
	$
\end{definition}
\noindent We can now describe Stern's identification scheme \\ \\
\cadre{
	\begin{center}
		Stern's single round Identification scheme $\IS_{\Stern}(\lambda,G) = (\ISKeygen,\ISProver = (P_1,P_2),\ISCheck,G;M,C = \{1,2,3\},R,n = 3).$
	\end{center}
	\textbf{Initialization. }$\ISKeygen(1^\lambda) : \Hm \Unif \FR^{(n-k),n}, \ev \Unif S_w, \sv \eqdef \ev\transpose{\Hm} \textrm{ return } pk = (\Hm,s), \ sk = e,$ where $n,k,w$ depend on the security parameter $\lambda$. \\ \\
	\textbf{Interaction. }$P_1: \sigma \Unif \RP^{[n]}, \yv \leftarrow \zo^n$. Let $\sv' \eqdef \yv\transpose{\Hm}$.   Let also $z_1 \eqdef (\sigma || \sv')$ ; $z_2 \eqdef \sigma(y)$ ; $z_3 \eqdef \sigma(\yv \oplus \ev).$ Send $(x_1,x_2,x_3) \eqdef \left(G(z_1),G(z_2),G(z_3)\right)$ to the verifier. \\
	$V : c \Unif \{1,2,3\}$, send $c$ to the prover. \\
	$P_2 : $ send $z_{c'}$ for the two values $c'$ different from  $c$. \\ \\
	\textbf{Verification. }$\ISCheck(1,(z_2,z_3)) = 1 $ iff. $|z_2 + z_3|_H = w.$ \\
	$\ISCheck(2,(z_1 \eqdef (\sigma,\sv'),z_3)) = 1 $ iff. $\sigma^{-1}(z_3)\transpose{\Hm} = \sv \oplus \sv'.$ \\
	$\ISCheck(3,(z_1 \eqdef (\sigma,\sv'),z_2)) = 1 $ iff. $\sigma^{-1}(z_2)\transpose{\Hm} = \sv'.$ \\
} $ \ $ \\ \\
One can check completeness. Indeed, in the honest case:
\begin{enumerate} \setlength\itemsep{-0.1em}
	\item $|z_2 + z_3|_H= |\sigma(\yv) + \sigma(\yv \oplus \ev)|_H = |\sigma(e)|_H = w.$
	\item $\sigma^{-1}(\sigma(\yv \oplus \ev))\transpose{\Hm} = \yv \transpose \Hm + \ev \transpose \Hm.$
	\item $\sigma^{-1}(\sigma(\yv))\transpose{\Hm} = \yv \transpose \Hm .$
\end{enumerate}
Moreover, suppose one constructs a triplet $z_1 = (\sigma,\sv'),z_2,z_3$ that passes the $3$ checks. We show how to easily construct a vector $\ev$ such that $\ev \transpose{\Hm} = \sv$ and $|\ev|_H = w$. Indeed, consider the vector $\ev = \sigma^{-1}(z_2 \oplus z_3)$. Using the second and third checks, we have 
$ \ev \transpose{\Hm} = \sigma^{-1}(z_2 \oplus z_3)\transpose{\Hm} = \sv \oplus \sv' \oplus \sv' = \sv.$
Also, $|\ev|_H = |\sigma^{-1}(z_2 \oplus z_3)|_H = |z_2 \oplus z_3|_H = w.$ This means we immediately have
\begin{align}\label{Equation:Last}
	QADV_{\IS_{\Stern}(\lambda,G)}^{3\textrm{-}sp}(t) = QAdv_{(n(\lambda),k(\lambda),w(\lambda))}^{\textup{SD}}(t).\end{align}
The above equality is exactly the kind of relations we need in order to prove the quantum security of the Fiat-Shamir transform of identifications schemes and hence of resulting signature schemes. Using Theorem \ref{Theorem:ParallelRepetition}, we immediately have 

\begin{proposition}[Quantum security of the Fiat-Shamir transform for the parallel repetition of Stern's identifications scheme] \label{Proposition:Stern}
	$$\E_{\substack{\hh \Unif \RF^{M^n}_C \\ G \Unif \RF^{R}_M}} \left[QADV_{\FS{{\IS^{\otimes r}_{\Stern}(\lambda,G)}}}(t,q_\hh,q_G)\right] \le QAdv_{(n(\lambda),k(\lambda),w(\lambda))}^{\textup{SD}}(O(t)) + O\left(\frac{q_{\hh}^2 2^r}{3^r}\right) + O\left(\frac{q_{G}^3}{|M|}\right).$$
\end{proposition}

From this proposition, we see that we can take $r = \frac{2\lambda}{\log_2(2/3)}$and $|M| = 3\lambda$ to get $\lambda$ bits of quantum security. 

\end{document}